\documentclass[11pt,onecolumn]{article}
\usepackage{times}
\usepackage{amsmath}
\usepackage{amssymb}
\usepackage{xspace}
\usepackage{theorem}
\usepackage{graphicx}
\usepackage{subfig}
\usepackage{epsfig}
\usepackage{ifpdf}
\usepackage{url,hyperref}
\usepackage{latexsym}
\usepackage{euscript}
\usepackage{xspace}
\usepackage{color}
\usepackage{makeidx}
\usepackage{picins,wrapfig}
\usepackage{stackrel}
\usepackage{amscd}
\usepackage[all]{xy}
\usepackage{multirow}
\usepackage{lineno}

\long\def\remove#1{}

\newtheorem{theorem}{Theorem}[section] % section
\newtheorem{lemma}[theorem]{Lemma}
\newtheorem{claim}[theorem]{Claim}

\newtheorem{proposition}[theorem]{Proposition}

\newenvironment{proof}{{\em Proof:}}{\hfill{\hfill\rule{2mm}{2mm}}}

\newcommand {\mm}[1] {\ifmmode{#1}\else{\mbox{\(#1\)}}\fi}

\newcommand{\eps}{{\varepsilon}}
\newcommand{\reals}	{{\rm I\!\hspace{-0.025em} R}}
\newcommand{\sphere}	{\mathbf{S}} %{\mathbb{S}}
\newcommand{\BB}			{\mathrm {\mathbb{B}}}
\newcommand{\Z}			{\mathrm {\mathbb{Z}}}

\newcommand{\MM}		{\mathrm {\mathbb{M}}}
\newcommand{\HH}		{\mathrm {\sf{H}}}
\newcommand{\redHH}	{\tilde{\mathrm {\sf{H}}}}
\newcommand{\X}			{\mathbb{X}}
\newcommand{\cirD}		{\mathring{D}}

\newcommand{\A}			{\mathrm {\mathbb{A}}}

\newcommand{\etal}      {et al.\@\xspace}

\newcommand{\inD}		{m}
\newcommand{\amD}		{d}

\newcommand{\M}		{{\sf M}}
\newcommand{\N}		{{\sf N}}
\newcommand{\nerv} {{\cal N}}

\newcommand{\rank}		{\mm {\rm rank}}

\newcommand{\Rips}	{{\cal R}}

\newcommand {\myparagraph}[1] {\vspace{0.15in} \noindent {\bf {#1}}}

\newcommand{\clX}		{\mathcal{X}}

\begin{document}

\title{Dimension Detection by Local Homology}

\author{
Tamal K. Dey\thanks{
Department of Computer Science and Engineering,
The Ohio State University, Columbus, OH 43210, USA.
Email: {\tt tamaldey@cse.ohio-state.edu}}
\quad\quad
Fengtao Fan\thanks{
Department of Computer Science and Engineering,
The Ohio State University, Columbus, OH 43210, USA.
Email: {\tt fanf@cse.ohio-state.edu}}
\quad\quad Yusu Wang\thanks{
Department of Computer Science and Engineering,
The Ohio State University, Columbus, OH 43210, USA.
Email: {\tt yusu@cse.ohio-state.edu}}
}
\date{}
\maketitle
\thispagestyle{empty}
\setcounter{page}{0}

\begin{abstract}
Detecting the dimension of a hidden manifold 
from a point sample has become an important
problem in the current data-driven era. 
Indeed, estimating the shape dimension is often the first step 
in studying the processes or phenomena  associated to the data. 
Among the many dimension detection algorithms proposed
in various fields, a few can provide theoretical guarantee on the 
correctness of the estimated dimension. However, the correctness 
usually requires certain regularity of the input: 
the input points are either uniformly randomly sampled in a 
statistical setting, or they form the so-called $(\eps,\delta)$-sample 
which can be neither too dense nor too sparse. 

Here, we propose a purely topological technique to detect dimensions.
Our algorithm is provably correct and works under
a more relaxed sampling condition: 
we do not require uniformity, and we also allow Hausdorff noise. 
Our approach detects dimension by determining 
local homology. The computation of this topological structure is much 
less sensitive to the local distribution of points, 
which leads to the relaxation of the sampling conditions. 
Furthermore, by leveraging various developments in computational topology, 
we show that this local homology at a point $z$ can be computed \emph{exactly} 
for manifolds using Vietoris-Rips complexes whose vertices are confined 
within a local neighborhood of $z$. 
We implement our algorithm and demonstrate the accuracy and 
robustness of our method using both synthetic and real data sets. 
\end{abstract}

\newpage

%\linenumbers
\section{Introduction}

% dimension estimation useful
% statistical setting
% computational geometry related work
% relevant work in local homology estimation
% our contribution

A fundamental problem in the current data-centric era is to estimate various qualitative structures from input data. Very often, the data is represented as a set of points sampled from a hidden domain. % and this hidden domain is embeded in an ambient Euclidean space. 
In particular, recent years have witnessed tremendous interest and progress in the field of \emph{manifold learning}, where the hidden domain is assumed to be a manifold $\M$ embedded in an ambient Euclidean space $\reals^d$. 
The intrinsic dimension of the manifold $\M$ is one of the simplest, yet still very important, quantities that one would like to infer from input data. 
Indeed, the dimension of $\M$ reflects the degree of freedom of the dynamic process that generates the data, and/or the number of variables necessary to describe the hidden domain. 
Hence, its estimation is crucial to our understanding of the processes or phenomena associated to the data. 
%, %and much research has been done in this direction; see e.g, surveys \cite{}. 

In this paper, we present an algorithm to estimate the intrinsic dimension of a manifold $\M$ from a set of noisy point samples $P \subset \reals^d$ \emph{on and around} $\M$. Our algorithm is based on the topological concept of local homology which was first investigated by Bendich \etal{} in the discrete setting
~\cite{Bendich:2007}. We show that our estimation is provably correct under appropriate sampling conditions and choice of parameters. 

\myparagraph{Related work.}
The problem of dimension estimation has been studied in 
various fields including pattern recognition, artificial intelligence and machine learning; see e.g., surveys \cite{Cam03,VD95}. If the domain of interest is linear, then the principal component analysis (PCA) \cite{J02} is perhaps the most popular method to estimate its dimension. However, PCA fails for non-linear domains and the curvature of the domain tends to cause PCA to overestimate the dimension. Fukunaga and Olsen pioneered the idea of using a local PCA applied to points within small neighborhoods for the non-linear case \cite{FO71}, and several variants have been developed along this direction \cite{BS98,LMR11}. In particular, Little \etal{} developed a multi-scaled version of the local PCA idea \cite{LMR11} that can achieve certain guarantee for points 
possibly corrupted with Gaussian noise, but
uniformly sampled from a hidden manifold.
A different approach estimates the manifold dimension based on the growth rate of the volume (or some analog of it) of an intrinsic ball~\cite{CV02,FSA07,GP83,HA05,PBJD79}. 
Both types of approaches above usually work in the statistical setting, where the input points are assumed to be sampled from some probabilistic distribution whose support is concentrated on the hidden manifold. 
%We remark that dimension estimation is also related to the problem of non-linear dimensionality reduction \cite{}. 

In the computational geometry community, Dey \etal{} \cite{DGGZ03} provided the first provably correct approach to estimate the dimension of a manifold $\M$ from a so-called $(\eps,\delta)$-sample of $\M$, which enforces a regularity of the point samples by requiring that these points are both $\eps$-dense and $\delta$-sparse. 
Their approach requires constructing the Voronoi diagram for input points, the computational cost of which becomes prohibitive when the ambient dimension is high. 
Requiring the same $(\eps,\delta)$-sampling condition from input points, Giesen and Wagner \cite{GW03} introduced the so-called adaptive neighborhood graph, and then locally fit (approximately) the best affine subspace under the $L_\infty$ norm to each sample point $p$ and its neighbors in this graph. The time complexity of their algorithm is exponential only in the intrinsic dimension and the detected  dimension is correct for appropriate parameters. 
Cheng \etal{} improved this result by applying a local PCA to each sample point and its neighbors in the adaptive neighborhood graph \cite{CWW05}. They also showed that a small amount of Hausdorff noise (of the order $\eps^2$ times the local feature size) and a sparse set of outliers can be tolerated in the input points. 
More recently, Cheng \etal \cite{CC09} proposed an algorithm to estimate dimension by detecting the so-called slivers. This algorithm works in a statistical setting, and assumes that the input points are sampled from the hidden manifold using a Poisson process without noise. 

In this paper we develop a dimension-detection method based on the topological concept of local homology. 
The idea of using local homology to understand spaces from sampled points was first proposed by Bendich \etal{}  \cite{Bendich:2007}. Specifically, they introduced multi-scale representations of local homology to infer on stratified spaces, and developed algorithms to compute these representations using the weighted Delaunay triangulation. 
This line of work was further developed in \cite{BWM12} where the so-called local homology transfer was proposed to cluster points from different strata. In a recent paper~\cite{bei:2012}, Skraba and Wang proposed to approximate the multi-scale representations of local homology using families of Rips complexes. 
Rips complexes are more suitable than the Delaunay triangulations for points sampled from \emph{low} dimensional compact sets embedded in \emph{high} dimensional space and have attracted much attention in topology inference \cite{ALS11,Chazal:2008,bei:2012}. 

\myparagraph{Our results.} 
Given a smooth $\inD$-dimensional manifold $\M$ embedded in $\reals^d$, the local homology group $\HH(\M, \M - z)$ at a point $z \in \M$ is isomorphic to the reduced homology group of a $\inD$-dimensional sphere, that is $\HH(\M, \M-z) \cong \redHH(\sphere^\inD)$. 
Hence, given a set of noisy sample points $P$ of $\M$, we aim to detect the dimension of $\M$ by estimating $\HH(\M, \M-z)$ from $P$. 
Specifically, we assume that $P$ is an $\eps$-sample\footnote{Note that this definition of $\eps$-sample allows points in $P$ to be $\eps$ distance off the manifold $\M$. Our $\eps$-sampling condition is with respect to the \emph{reach} of $\M$ while that used in \cite{CC09,CWW05,DGGZ03,GW03} is with respect to local feature size and thus adaptive.} of $\M$ in the sense that the Hausdorff distance between $P$ and $\M$ is at most $\eps$. 
Our main result is that by inspecting two nested neighborhoods around a sample point $p \in P$ and considering certain relative homology groups computed from the Rips complexes induced by points within these neighborhoods, one can recover the local homology \emph{exactly}; see Theorem \ref{thm:mainthm}. 
This in turn provides a provably correct dimension-detection algorithm for an $\eps$-sample $P$ of a hidden manifold $\M$ when $\eps$ is small enough.
%compared to the convexity radius of $\M$). 

Compared with previous provable results in \cite{CC09,CWW05,DGGZ03,FSA07,GW03,LMR11}, our theoretical guarantee on the estimated dimension is obtained with a more relaxed sampling condition on $P$.
% \footnote{We remark that our $\eps$-sampling condition is with respect to the \emph{reach} of $\M$ while that used in \cite{CC09,CWW05,GW03} is with respect to local feature size and thus adaptive.}. 
Specifically, there is no uniformity requirement for the sample points $P$, which was required by all previous dimension-estimation algorithms with theoretical guarantees: either in the form of a uniform random sampling in the statistical setting \cite{CC09,FSA07,LMR11} or the $(\eps,\delta)$-sampling in the 
deterministic setting \cite{CWW05,DGGZ03,GW03}. We also allow larger amount of noise ($\eps$ vs. $\eps^2$ as in \cite{CC09}). 
Such a relaxation in the sampling condition is primarily made possible by considering the topological information, which is much less sensitive to the distribution of points compared to the approaches based on local fitting. 

In Section \ref{exp}, we provide preliminary experimental results of our algorithm on both synthetic and 
real data. For synthetic data our method detects the right dimension
robustly. For real data some of which are laden with high noise
and undersampling, not all points return the correct dimension. But,
taking advantage of the fact that local homology is trivial in all but 
zero and intrinsic dimension of the manifold, we can eliminate 
most false positives and estimate the correct dimension from appropriately chosen points. 
%from points with well sampled neighborhoods. 
%See Section~\ref{exp} for details. 

Finally, we remark that similar to the recent work in \cite{bei:2012}, our computation of local homology uses the Rips complex, which is much easier to construct than the ambient Delaunay triangulation as was originally required in \cite{Bendich:2007}. Different from \cite{bei:2012}, we aim to compute $\HH(\M, \M-z)$ \emph{exactly} for the special case when $\M$ is a manifold, while the work in \cite{bei:2012} \emph{approximates} the multiscale representations of local homology (the persistence diagram of certain filtration) for more general compact sets. We also note that, unlike~\cite{bei:2012} our algorithm
operates with Rips complexes that span vertices within a local
neighborhood, thus saving computations. The goals from these two works are somewhat complementary and the two approaches address different technical issues.

\section{Preliminaries and Notations}

\paragraph{Manifold and sample.}
Let $\M$ be a compact smooth $\inD$-dimensional manifold without boundary embedded
in an Euclidean space $\reals^d$. The {\em reach} $\rho(\M)$ is the
minimum distance of any point in ${\M}$ to its medial axis.
A finite point set $P\subset \reals^d$ is an
\emph{$\eps$-sample} of $\M$ if every point $z\in \M$ satisfies
$d(z,P)\leq \eps$ and every point $p \in P$ satisfies $d(p, \M) \leq \eps$; in other words, the \emph{Hausdorff distance} between $P$ and $\M$ is at most $\eps$.

\myparagraph{Balls.}
An Euclidean closed ball with radius $r$ and center $z$ is denoted
$B_r(z)$. The open ball with the same center and radius is denoted
$\mathring{B}_r(z)$ and its complement 
$\reals^d\setminus \mathring{B}_r(z)$
is denoted $B^r(z)$.

\myparagraph{Homology.} We denote the $i$-th dimensional homology
group of a topological space $X$ as $\HH_i(X)$. We drop $i$ and write
$\HH(X)$ when a statement holds for all dimensions. 
We mean by $\HH(X)$ the singular homology
if $X$ is a manifold or a subset of $\reals^d$, and simplicial
homology if $X$ is a simplicial complex. Both homologies are
assumed to be defined with $\mathbb{Z}_2$ coefficients. We make 
similar assumptions to denote the relative homology
groups $\HH(X,A)$ for $A\subseteq X$. Notice that both $\HH(X)$ and $\HH(X,A)$
are vector spaces because they are defined with $\Z_2$ coefficients. 
The following two known results will be used several times in this paper. 

\begin{proposition}[\cite{Chazal:2008}]
Let $\HH(A)\rightarrow \HH(B)\rightarrow \HH(C)\rightarrow \HH(D)\rightarrow \HH(E)
\rightarrow \HH(F)$ be a sequence of homomorphisms.
If $\rank (\HH(A)\rightarrow \HH(F))=\rank(\HH(C)\rightarrow \HH(D))=k$,
then $\rank (\HH(B)\rightarrow \HH(E))=k$.
\label{rank}
\end{proposition}
\vspace*{-0.2in}\begin{proposition}[Steenrod-five lemma (Lemma 24.3 in ~\cite{Munk75})]
Suppose we have the commutative diagram of homology groups and homomorphisms:
$$
\xymatrix
{
&\HH_i(A) \ar[r] \ar[d]^{f_1}
& \HH_i(X) \ar[r] \ar[d]^{f_2}
& \HH_i(X,A) \ar[r] \ar[d]^{f_3}
& \HH_{i-1}(A) \ar@{->}[r] \ar[d]^{f_4}
& \HH_{i-1}(X) \ar[d]^{f_5}
&
\\
&\HH_i(B) \ar[r]
& \HH_i(Y) \ar[r]
& \HH_i(Y,B) \ar[r]
& \HH_{i-1}(B) \ar[r]
& \HH_{i-1}(Y) 
&
}
$$
where the horizontal sequences are exact. If $f_1,f_2,f_4$, and $f_5$
are isomorphisms, so is $f_3$.
\label{fivelem}
\end{proposition}

%\begin{itemize}
%\item $\mathring{B}_r(z)$ \quad open Euclidean ball with radius $r$ at point p;
%\item $B_r(z)$ \quad closed Euclidean ball with radius $r$ at point p;
%\item $B^r(z)=\mathring{B}_r(z)^{c}$ \quad complement of $\mathring{B}_r(z)$;
%\item $\M$ \quad \quad smooth sub-manifold of dimension $n$ in ${\cal R}^d$;
%\item $\mathring{B}^{\cal G}_r(z)$ \quad open geodesic ball with radius $r$ at point $p \in {\M}$;
%\item $B^{\cal G}_r(z)$ \quad closed geodesic ball with radius $r$ at point $p \in {\M}$;
%\item ${ X} = \{ x_i \}_{i=1}^{N}$ \quad a sampling point cloud of $\M$;
%\item ${ X}_\alpha = \cup_{i=1}^{N} B_\alpha(x_i)$ \quad the union of balls centered at $x_i$'s with radius $\alpha$;
%\item $\rho$	\quad positive reach of $\M$;
%\item $\mathbf{inj}({\M})$ \quad injective radius of $\M$;
%\end{itemize}

\paragraph{Overview of approach. }
We are given an $\eps$-sample $P=\{p_i\}_{i=1}^n$ of a compact smooth $\inD$-manifold $\M$ embedded in $\reals^\amD$. 
However, the intrinsic dimension $\inD$ of $\M$ is not known, and our goal is to estimate $\inD$ from the point sample $P$. 
Note that for any point $z \in \M$, we have that $\HH(\M, \M - z) 
\cong \redHH(\sphere^\inD)$ where $\redHH(\cdot)$ denotes the
reduced homology. Thus $\rank(\HH_i(\M,\M-z))=1$ if and only if
$i= m$. Hence, if we can compute the rank of $\HH_i(\M, \M-z)$ for every $i$, then we can recover the dimension of $\M$. 
This is the approach we will follow. 
In Section \ref{sec:offsets}, we first relate $\HH(\M, \M-z)$ with the topology of the offset of the point set $P$. This requires us to inspect the deformation retraction from the offset to $\M$ carefully. 
The relation to the offset, in turns, allows us to provably recover the rank of $\HH(\M, \M-z)$ using the so-called Vietoris Rips complex, which we detail in Section \ref{sec:Rips}. One key ingredient here is to use only local neighborhoods of a sample point to obtain the estimate. 
First, in Section \ref{sec:localhom}, we derive several technical results to prepare for the development of our approach in Section \ref{sec:offsets} and \ref{sec:Rips}. 

%%%%%%%%%%%%%%%%%%%%%%%%%%%%%%%%%
\section{Local Homology of $\M$ and its Offsets}
\label{sec:localhom}
%%%%%%%%%%%%%%%%%%%%%%%%%%%%%%%%%

%In this section we establish that the local homology at a point
%in the manifold $\M$ is same as that of an offset to its $\eps$-sample
%$P=\{p_i\}_{i=1}^n$.
%The following result is known and becomes useful for the development
%of our results.

\paragraph{Local homology $\HH(\M,\M-z)$.}
In this section, we develop a few results that we use later. 
First, we relate the target local homology groups $\HH(\M, \M-z)$ to 
some other local homology which becomes useful later for connecting to
the local homology of Rips complexes that are ultimately used in the algorithm. 
We start by quoting the following known result: 
\begin{proposition}[\cite{Dey06}]
Let $B_r(p)$ be a closed Euclidean ball so that it
intersects the $\inD$-manifold ${\M}$ in more than one point.
If $r<\rho(\M)$, then ${\M}\cap B_r(p)$ is a closed topological $\inD$-ball. 
\label{ball-intersect}
\end{proposition}

%Assume that $z\in \M$ and $D \subset \M$  is a closed topological $\inD$-ball that contains $z$ in its interior $\cirD$.  
%%let $D={\M}\cap B_r(z)$ be a closed topological $k$-ball 
%%where $\cirD$ denote the interior of $D$. 
%Our first observation is that the inclusion in the pair
%$(\M, \M-\cirD)\stackrel{i}{\hookrightarrow} (\M, \M-z)$ induces
%an isomorphism at the homology level, that is: 
\begin{proposition}
Let $D \subset \M$ be a closed topological $\inD$-ball from the $\inD$-manifold $\M$, and $z\in \M$ a point contained in the interior $\cirD$ of $D$. 
Then $\HH(\M,\M-\cirD)\stackrel{i_*}{\rightarrow} \HH(\M, \M-z)$ is an isomorphism.
\label{point-ball-lem}
\end{proposition}
\begin{proof}
Consider the following diagram where the two horizontal
sequences are exact and all vertical maps are induced by inclusions:
$$
\xymatrix
{
&\HH_i(\M-\cirD) \ar[r] \ar[d]^{i'_*}
& \HH_i(\M) \ar[r] \ar[d]^{\cong}
& \HH_i(\M,\M-\cirD) \ar[r] \ar[d]^{i_*}
& \HH_{i-1}(\M-\cirD) \ar@{->}[r] \ar[d]^{i'_*}
& \HH_{i-1}(\M) \ar[d]^{\cong}
&
\\
&\HH_i(\M-z) \ar[r]
& \HH_i(\M) \ar[r]
& \HH_i(\M,\M-z) \ar[r]
& \HH_{i-1}(\M-z) \ar[r]
& \HH_{i-1}(\M) 
&
}
$$
As all vertical homomorphisms are induced by inclusions, the above
diagram commutes, see Theorem 5.8 in Rotman~\cite{Rotman}.
Consider the inclusion $(\M-\cirD) \stackrel{i'}{\hookrightarrow} (\M-z)$.
Since $D$ is a closed topological ball, $\M-z$ deformation retracts
to $\M-\cirD$. The inclusion $i'$ is a homotopy inverse
of the retraction $(\M-z) \rightarrow (\M-\cirD)$ and hence $i'_*$ is an 
isomorphism. Since the first, second, fourth and fifth vertical
homomorphisms in the above diagram are isomorphisms, $i_*$ is also an
isomorphism by Proposition~\ref{fivelem}.
\end{proof}

\vspace*{0.08in}We can extend Proposition~\ref{point-ball-lem} a little further. See Appendix \ref{appendix:ball-ball-lem} for the proof. 
\begin{proposition}
Let $D_1$ and $D_2$ be two closed topological balls containing $z$ 
in the interior where
$D_1\subseteq D_2 \subseteq \M$. 
The inclusion-induced homomorphisms $i_*'$ and $i_*$ in the following sequence are isomorphisms:  \\
\hspace*{1.5in}$\HH(\M,\M-\cirD_2)\stackrel{i'_*}{\rightarrow} 
\HH(\M,\M-\cirD_1)\stackrel{i_*}{\rightarrow} (\M,\M-z). $
%are isomorphisms.
\label{ball-ball-lem}
\end{proposition}
%\begin{proof}
%We only need to show that $i_*'$ is an isomorphism
%as Proposition~\ref{point-ball-lem} proves it for $i_*$.
%Since the inclusion induced homomorphisms
%$j_*: \HH(\M,\M-\cirD_2)\rightarrow \HH(\M,\M-z)$ 
%and $i_*: \HH(\M,\M-\cirD_1)\rightarrow \HH(\M,\M-z)$ are isomorphisms
%by Proposition~\ref{point-ball-lem} and $j_*=i_*\circ i'_*$,
%we have that $i'_*$ is an isomorphism as well. 
%\end{proof}

\myparagraph{Local homology of the offset.}
%\vspace*{-0.2in}
%\vspace*{0.1in}\noindent{\bf Local homology of the offset.~}
Later we wish to relate the local homology $\HH(\M,\M-z)$ at a point $z$ to the local
homology of an $\alpha$-offset of an $\eps$-sample $P=\{p_i\}_{i=1}^n$, defined as 
%Specifically, for $\alpha <\rho-\eps$, let
$${\X}_\alpha = \cup_{i=1}^{n} B_\alpha(p_i), 
\mbox{ the union of balls centered at every $p_i$ with radius $\alpha$}.
$$
For this, we will need a map to connect the two spaces, which is 
provided by the following projection map: 
%Consider the projection map 
$$\pi_{\alpha} : { \X}_\alpha \rightarrow {\M} \mbox{ given by } 
x \mapsto \mathrm{argmin}_{z\in {\M}} d(x,z).
$$
Choose $\alpha < \rho(\M)-\eps$. Since $P$ is an $\eps$-sample, no point
of $\X_\alpha$ is $\rho(\M)$ or more
away from $\M$. This means that no point
of the medial axis of $\M$ is included in $\X_\alpha$. Therefore,
the map $\pi$ is well defined.
Furthermore, by the following result of \cite{Niyogi:2008}, $\pi$ is a deformation retraction for appropriate choices of parameters. In fact, under this projection map, the pre-image of a point has a nice structure (star-shaped). 
\begin{proposition}[pp.22, \cite{Niyogi:2008}] If $P$ is an $\eps$-sample of $\M$ with reach $\rho =\rho(\M)$
where $0<\eps < (3-\sqrt{8})\rho$ and 
$\alpha \in (\frac{(\eps + \rho)-\sqrt{\eps^{2} + \rho^{2} - 6\eps\rho}}{2}, 
\frac{(\eps + \rho)+\sqrt{\eps^{2} + \rho^{2} - 6\eps\rho}}{2})$,
then, for any $x\in \pi_\alpha^{-1}(z)$, the segment
$xz$ lies in $\pi^{-1}_\alpha(z)$.
\label{fiber}
\end{proposition}

For convenience denote 
$\theta_1=\frac{(\eps + \rho)-\sqrt{\eps^{2} + \rho^{2} - 6\eps\rho}}{2}$ and 
%\mbox{ and }
$\theta_2= \frac{(\eps + \rho)+\sqrt{\eps^{2} + \rho^{2} - 6\eps\rho}}{2}$
and observe that $\eps\leq \theta_1$ and
$\theta_2 \leq \rho(\M)-\eps$ for $\eps,\rho>0$. We have:

\begin{proposition}
Let $0<\eps < (3-\sqrt{8})\rho(\M)$ 
and  $\theta_1 \leq \alpha \leq 
\theta_2$. 
%$\eps\leq\alpha\leq \rho(\M)-\eps$.
Let $\A_{\alpha} = \pi^{-1}_\alpha({\N})$ where ${\N} \subseteq {\M}$ may be  
either an open or a closed subset. 
Then $\pi_\alpha: \A_{\alpha} \rightarrow {\N}$ 
is a retraction and
${\N}$ is a deformation retract of $\A_\alpha$.
\label{deform}
\end{proposition}
\begin{proof}
%First observe that, $\M\subset \X_\alpha$ because 
%$\alpha \geq\theta_1\geq \eps$ and
%$P$ is an $\eps$-sample. 
Notice that due to Proposition~\ref{fiber}, $\pi^{-1}_\alpha(z)$ is star shaped meaning that every point $x \in \pi^{-1}_\alpha(z)$ has the segment $xz$ lying in $\pi^{-1}_\alpha(z)$. It follows that ${\N}\subseteq \A_{\alpha}$ and there exists a straight line deformation retraction 
$F: \A_{\alpha}\times I \rightarrow \A_{\alpha}$ defined as $F(x,t) = (1-t)x + t\pi(x)$.
The proposition then follows.
\end{proof}

Based on the above observation, the map 
$\pi_\alpha : (\X_\alpha, \A_{\alpha}) \rightarrow ({\M} , {\N})$ 
seen as a map on the pairs provides an
isomorphism at the homology level.

\begin{proposition} 
Let
$0<\eps < (3-\sqrt{8})\rho$ 
and  
$\theta_1 \leq \alpha \leq 
\theta_2$. 
The homomorphism 
$\pi_{\alpha *}: \HH(\X_\alpha, \A_{\alpha}) \rightarrow \HH({\M} , {\N})$
is an isomorphism.
\label{pair-iso}
\end{proposition}
\begin{proof}
The map $\pi_\alpha$ provides the following commutative diagram 
(Theorem 5.8, Rotman~\cite{Rotman}):
$$
\xymatrix
{
\ar[r]
&\HH_i(\A_{\alpha}) \ar[r] \ar[d]^{\pi_{\alpha*}}
& \HH_i(\X_\alpha) \ar[r] \ar[d]^{\pi_{\alpha*}}
& \HH_i(\X_\alpha, \A_{\alpha}) \ar[r] \ar[d]^{\pi_{\alpha*}} 
& \HH_{i-1}(\A_{\alpha}) \ar@{->}[r] \ar[d]^{\pi_{\alpha*}} 
& \HH_{i-1}(\X_\alpha) \ar@{->}[r] \ar[d]^{\pi_{\alpha*}} 
&
\\
\ar[r]
&\HH_i({\N}) \ar[r]  
& \HH_i({\M}) \ar[r] 
& \HH_i({\M}, {\N}) \ar[r]  
& \HH_{i-1}({\N}) \ar[r] 
& \HH_{i-1}({\M}) \ar[r] 
&
}
$$
The first, second, fourth, and fifth vertical maps are 
restrictions of $\pi_{\alpha*}$ and thus are all isomorphisms
by Proposition~\ref{deform}. It follows from
Proposition~\ref{fivelem} that the third vertical map 
is an isomorphism as well.
\end{proof}

\begin{proposition}
Let 
$0<\eps < (3-\sqrt{8})\rho$, 
and  
$\theta_1 \leq \alpha < \alpha' \leq 
\theta_2$.
Let ${\N}\subset {\N}' $ be two closed (or open) sets of ${\M}$,
and $\A_{\alpha}= \pi^{-1}_\alpha({\N})$ and $\A_{\alpha'}= \pi^{-1}_{\alpha'}({\N}')$. 
Denoting by $\mathrm{im} (\cdot)$ the image of a map, we have 
$$\mathrm{im}\left(\HH(\mathbb{X}_\alpha, \A_{\alpha}) \rightarrow \HH(\mathbb{X}_{\alpha'}, \A_{\alpha'})\right) \cong
 \mathrm{im}\left(\HH({\M}, {\N}) \rightarrow \HH({\M}, {\N}')\right). $$
\label{persistence-equiv}
\vspace*{-0.3in}
\end{proposition}
\begin{proof}
The projection maps $\pi_{\alpha}$ and $\pi_{\alpha'}$ (both being maps of pairs) 
%result in this commutative diagram: 
result in the following commutative diagram of pairs.
\[
\xymatrix{
(\mathbb{X}_\alpha, \A_{\alpha}) \ar[d]^{\pi_{\alpha}} \ar@{^{(}->}[r] & (\mathbb{X}_{\alpha'}, \A_{\alpha'}) \ar[d]^{\pi_{\alpha'}}\\
({\M}, {\N}) \ar@{^{(}->}[r] & ({\M}, {\N}')}
\]
This diagram induces a commutative diagram at homology level, 
where $\pi_{\alpha*}$ and $\pi_{\alpha'*}$ are isomorphisms by 
Proposition~\ref{pair-iso}.
The claim now is immediate by the 
Persistence Equivalence Theorem~\cite{EH09}, page 159. 
\end{proof}
%\section{Local homology}
%Let $p \in {\M}$ be a point. 
%The local homology at $p$ is the relative homology $\HH({\M}, {\M}-z)$. 
%If $0 \leq r < \mathbf{inj}(\M)$, then ${\M} - p$ deformation retracts to ${\M} - \mathring{B}^{\cal G}_r(z)$ through the geodesic paths within the ball ${B}^{\cal G}_r(z)$.
%Therefore, it follows that $\HH({\M}, {\M} - \mathring{B}^{\cal G}_r(z)) \cong \HH({\M}, {\M}-z)$.
%%The deformation retraction $\pi$ is essentially a deformation equivalence between the pairs $({ \X}_\alpha, \pi^{-1}({\M} - \mathring{B}^{\cal G}_r(z))$
%%and $({\M}, {\M} - \mathring{B}^{\cal G}_r(z))$. 
%%Thus, it holds that 
%According to the discussion in section 2, the following holds
%$$\HH({ \X}_\alpha, \pi^{-1}({\M} - \mathring{B}^{\cal G}_r(z))) \cong \HH({\M}, {\M} - \mathring{B}^{\cal G}_r(z))$$
%Furthermore, $$\HH({ \X}_\alpha, \pi^{-1}({\M} - \mathring{B}^{\cal G}_r(z))) \cong \HH({\M}, {\M} - z)$$
 
%%%%%%%%%%%%%%%%%%%%%%%%%%%%%
\section{Local Interleaving of Offsets}
\label{sec:offsets}
%%%%%%%%%%%%%%%%%%%%%%%%%%%%%

%In this section we interleave pairs of subsets of $\X_\alpha$ with those of
%the local neighborhoods of a point $p \in \X_\alpha$. 
%These local neighborhoods have local homology of the projected
%point $\pi(p)$ on $\M$ because of the deformation retractions. The interleaving
%allows us to relate the relative homology of a pair of subsets of
%$\X_\alpha$ with the local homology of the neighborhoods of $p$ and hence to
%that of $\pi(p)$. Since $\pi(p)$ plays an important
%role here, we use a special symbol $\bar{p}=\pi(p)$ for it.

Let $p\in P$ be any sample point.
We show how to obtain the local homology of the 
projected point $\pi(p)$ on $\M$
from pairs of $p$'s local neighborhoods in $\X_\alpha$. 
The results from the previous section already allow us to relate the local homology of the projected point $\pi(p)$ with the local homology of some local neighborhoods in $\X_\alpha$ (which are the pre-image of some sets in $\M$). 
We now use interleaving to relate them further to local neighborhoods 
that are intersection of $\X_\alpha$ with Euclidean balls. 
Since $\pi(p)$ plays an important role here, we use a special symbol $\bar{p}=\pi(p)$ for it.
For convenience, we introduce notations 
(see Figure~\ref{spaces-fig}): 
$$
\MM_{\alpha,\beta}=\pi_\alpha^{-1}(\mathring{B}_{\beta}(p) \cap {\M}), \,\,
{\MM}^{\alpha,\beta}= \X_\alpha - \MM_{\alpha,\beta},
\mbox{ and } 
\BB_{\alpha,\beta} = \mathring{B}_\beta(p)\cap \X_\alpha, \, \,
{\BB}^{\alpha,\beta}= \X_\alpha - \BB_{\alpha,\beta}.
$$
\begin{figure}[ht!]
\begin{center}
\input{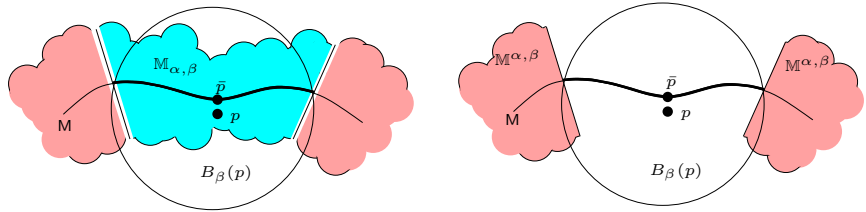}
\end{center}
\vspace*{-0.1in}
\caption{The spaces $\mathbb{M}_{\alpha,\beta}$ shown in cyan (left) and 
$\mathbb{M}^{\alpha,\beta}$ shown in pink (right).}
\label{spaces-fig}
\end{figure}

%We make the following simple observation that is used later.
The following simple observation follows from Propositions \ref{point-ball-lem},
\ref{ball-intersect}, and \ref{deform}. 
\begin{proposition}
Let $D_\beta=B_\beta(p)\cap \M$. 
For 
$0<\eps < (3-\sqrt{8})\rho$ , $\eps < \beta < \rho(\M)$ 
and  
$\theta_1 \leq \alpha \leq \theta_2$, 
the maps $\pi_{\alpha*}$ and $i_*$ are isomorphisms in the sequence: 
%\\ \hspace*{1.2in} 
$\HH(\X_\alpha,{\MM}^{\alpha,\beta})\stackrel{\pi_{\alpha*}}{\rightarrow}
\HH(\M,\M-\cirD_\beta)\stackrel{i_*}{\rightarrow}
\HH(\M,\M-\bar{p}). $
\label{projpair}
\end{proposition}

Now set $\delta=\alpha +3\eps$. 
Consider any $z\in \M$. 
Since any point $x\in \pi_\alpha^{-1}(z)$ resides
within a ball $B_\alpha(p_i)$ for some $p_i\in P$, we have that 
\begin{eqnarray}
d(x,z)=d(x,\pi(x))\leq d(x,\pi(p_i))\leq d(x,p_i) + d(p_i,\pi(p_i))  \leq \alpha+\eps=\delta - 2\eps.
\label{dist-nest}
\end{eqnarray} 
It follows that for any $\lambda\in(\eps,\rho(\M)-\delta)$
%\begin{eqnarray}
%\MM_{\alpha,\lambda}=\pi^{-1}_\alpha(\mathring{B}_{\lambda}(p) \cap {\M}) \subset 
%\pi^{-1}_\alpha(\mathring{B}_{\lambda+\eps}(\bar{p}) \cap {\M}) \subset
%\mathring{B}_{\lambda+\delta - \eps}(\bar{p}) \cap \X_\alpha \subset
%\mathring{B}_{\lambda+\delta}(p) \cap \X_\alpha=\BB_{\alpha,\lambda+\delta} .
%\label{eq1}
%\end{eqnarray}
%Now take a point $x \in B_\lambda(p) \cap \X_\alpha$. 
%Then $d(\pi(x) , p) \leq d(\pi(x) ,x) + d(x,p) \leq (\alpha+\eps) + \lambda$. 
%Therefore, 
%\begin{eqnarray}
%\BB_{\alpha,\lambda}=B_\lambda(p) \cap \X_\alpha \subset 
%\pi^{-1}_{\alpha}(B_{\lambda+\delta-2\eps}(p) \cap {\M})=\MM_{\alpha, \lambda+\delta-2\eps} \subset
%\MM_{\alpha,\lambda+\delta}.
%\label{eq2}
%\end{eqnarray}
we get the following inclusions(see Appendix~\ref{appendix:sec:offsets} for details):
%in Eqns~(\ref{eq1}) and~(\ref{eq2}) for any $\lambda\in(\eps,\rho(\M)-4\delta)$:
$$
%\pi^{-1}(\mathring{B}_{\lambda}(z) \cap {\M}) 
\MM_{\alpha,\lambda}
\subset \BB_{\alpha,\lambda+\delta}
\subset \MM_{\alpha,\lambda+2\delta}
\subset \BB_{\alpha,\lambda+3\delta}
\subset \MM_{\alpha, \lambda+4\delta}.
$$
Taking the complements, a new filtration in the reverse direction is generated:
$$
{\MM}^{\alpha,\lambda+4\delta}
\subset
{\BB}^{\alpha,\lambda+3\delta}
\subset
{\MM}^{\alpha,\lambda+2\delta}
\subset
{\BB}^{\alpha,\lambda+\delta}
\subset
{\MM}^{\alpha,\lambda}.
$$
Considering each space as a topological pair, the nested sequence becomes
\begin{eqnarray}
(\X_\alpha, {\MM}^{\alpha,\lambda + 4 \delta})
\subset
(\X_\alpha, {\BB}^{\alpha,\lambda+3\delta})
\subset
(\X_\alpha, {\MM}^{\alpha,\lambda + 2 \delta})
\subset
(\X_\alpha, {\BB}^{\alpha,\lambda+\delta})
\subset
(\X_\alpha, {\MM}^{\alpha,\lambda })
\label{eq2half}
\end{eqnarray}
Inclusion between topological pairs induces a 
homomorphism between their relative homology groups.
Therefore, the following relative homology sequence holds.
\begin{eqnarray}
\HH(\X_\alpha, {\MM}^{\alpha,\lambda + 4 \delta})
\rightarrow
\HH(\X_\alpha, {\BB}^{\alpha,\lambda+3\delta})
\rightarrow
\HH(\X_\alpha, {\MM}^{\alpha,\lambda + 2 \delta})
\rightarrow
\HH(\X_\alpha, {\BB}^{\alpha,\lambda+\delta})
\rightarrow
\HH(\X_\alpha, {\MM}^{\alpha,\lambda })
\label{eq3}
\end{eqnarray}

\noindent Let $\epsilon \leq \alpha' \leq \rho({\M}) - \epsilon$ and $\delta' = \alpha' + 3\eps$. Similar to sequence~(\ref{eq2half}), for any $\lambda' \in (\eps, \rho(\MM) - 4\delta')$ we have:  
\begin{eqnarray}
(\mathbb{X}_{\alpha'}, \mathbb{M}^{\alpha',\lambda' + 4\delta'}) 
\subset (\mathbb{X}_{\alpha'}, \mathbb{B}^{\alpha',\lambda' + 3\delta'})
\subset (\mathbb{X}_{\alpha'}, \mathbb{M}^{\alpha',\lambda' + 2\delta'})
\subset (\mathbb{X}_{\alpha'}, \mathbb{B}^{\alpha',\lambda' + \delta'})
\subset (\mathbb{X}_{\alpha'}, \mathbb{M}^{\alpha',\lambda'})
\label{eq4}
\end{eqnarray}

%Combining these two nested sequences in~(\ref{eq2half})
%and~(\ref{eq4}) together, we have the following sequence of inclusions 
%Then we have the following sequence of inclusions, for $\lambda \in (\eps, \rho(\MM) - \delta - 2\alpha')$, 
%\begin{eqnarray}
%(\mathbb{X}_\alpha, \mathbb{M}^{\alpha,\lambda + 2\alpha' + 4\delta}) \hookrightarrow
%(\mathbb{X}_\alpha, \mathbb{B}^{\alpha,\lambda + 2\alpha' + 3\delta}) \hookrightarrow
%(\mathbb{X}_\alpha, \mathbb{M}^{\alpha,\lambda + 2\alpha' + 2\delta}) \hookrightarrow
%(\mathbb{X}_{\alpha'}, \mathbb{M}^{\alpha',\lambda + 2\delta'}) \hookrightarrow
%(\mathbb{X}_{\alpha'}, \mathbb{B}^{\alpha',\lambda + \delta'}) \hookrightarrow
%(\mathbb{X}_{\alpha'}, \mathbb{M}^{\alpha',\lambda})
%\label{eq8}
%\end{eqnarray}
The stated range of $\lambda,\lambda'$ is valid if 
$\alpha,\alpha' < \frac{\rho(\M)-13\eps}{4}$. We also need
$\theta_1\leq \alpha,\alpha'$. These two conditions are
satisfied for $\eps < \frac{\rho(\M)}{22}$. 
%Note that when $0<\eps < \frac{\rho(\M)}{22}$, one has that $\theta_1 < \frac{\rho(\M)-13\eps}{4} < \theta_2$.
Let $\theta_2'= \frac{\rho(\M)-13\eps}{4}$.
\begin{proposition}
Let $0<\eps < \frac{\rho(\M)}{22}$, and  
$\theta_1 \leq \alpha \leq \alpha' \leq \theta_2'$. 
Set $\delta=\alpha + 3\eps$ and $\delta' = \alpha' + 3\eps$.   
For $\eps < \lambda'< \rho(\M)-4\delta'$ and $\lambda \ge  \lambda' + 2(\alpha'-\alpha)$, we have, 
\begin{eqnarray}
\mathrm{im}\left(\HH(\X_\alpha, {\BB}^{\alpha, \lambda+3\delta}) \rightarrow \HH(\X_{\alpha'}, {\BB}^{\alpha',\lambda'+\delta'})\right) \cong \HH({\M}, {\M} - \bar{p}). 
\label{eqn:im:prop}
\end{eqnarray}
In particular, 
$
\mathrm{im}\left(\HH(\X_\alpha, {\BB}^{\alpha,\lambda+3\delta}) \rightarrow \HH(\X_\alpha, {\BB}^{\alpha,\lambda+\delta})\right) \cong \HH({\M}, {\M} - \bar{p}). 
$
\label{im:prop}
\end{proposition}
\begin{proof}
Due to our choice of parameters, we have that $\lambda + 2\delta \ge \lambda' + 2\delta'$. 
From Eqn (\ref{eq2half}) and (\ref{eq4}), we obtain the following sequence of homomorphisms induced by inclusions: 
\begin{eqnarray*}
\HH(\mathbb{X}_\alpha, \mathbb{M}^{\alpha,\lambda + 4\delta}) \rightarrow
\HH(\mathbb{X}_\alpha, \mathbb{B}^{\alpha,\lambda + 3\delta}) \rightarrow
\HH(\mathbb{X}_\alpha, \mathbb{M}^{\alpha,\lambda + 2\delta}) \rightarrow \\
\HH(\mathbb{X}_{\alpha'}, \mathbb{M}^{\alpha',\lambda' + 2\delta'}) \rightarrow
\HH(\mathbb{X}_{\alpha'}, \mathbb{B}^{\alpha',\lambda' + \delta'}) \rightarrow
\HH(\mathbb{X}_{\alpha'}, \mathbb{M}^{\alpha',\lambda'}). 
\end{eqnarray*}
We first show
\begin{eqnarray}
\mathrm{im} \left( \HH(\mathbb{X}_\alpha, \mathbb{M}^{\alpha,\lambda + 4\delta}) \rightarrow \HH(\mathbb{X}_{\alpha'}, \mathbb{M}^{\alpha',\lambda'})\right)
\cong \mathrm{im} \left( \HH(\mathbb{X}_\alpha, \mathbb{M}^{\alpha,\lambda + 2\delta}) \rightarrow (\mathbb{X}_{\alpha'}, \mathbb{M}^{\alpha',\lambda' + 2\delta'})\right) 
\cong \HH({\M}, {\M} - \bar{p}). 
\label{eq10}
\end{eqnarray}
Consider the following commutative diagram where 
$\pi_\alpha$ and $\pi_{\alpha'}$ are seen as maps on pairs: 
%\begin{wrapfigure}{l}{2in}
\[
\xymatrix{
(\mathbb{X}_\alpha, \mathbb{M}^{\alpha, \lambda + 4\delta}) \ar[d]^{\pi_{\alpha}} \ar@{^{(}->}[r]
& (\mathbb{X}_{\alpha'}, \mathbb{M}^{\alpha', \lambda'}) \ar[d]^{\pi_{\alpha'}}\\
({\M}, {\M} - \mathring{D}_{\lambda + 4\delta}) \ar@{^{(}->}[r] & ({\M}, {\M}-\mathring{D}_{\lambda'}))}
\]
%\end{wrapfigure}
where $D_\beta=B_\beta(p)\cap \M$.
By Proposition~\ref{ball-ball-lem},
we have 
$$\mathrm{im}\left(\HH({\M}, {\M} - \mathring{D}_{\lambda + 4\delta}) \rightarrow
\HH({\M}, {\M}-\mathring{D}_{\lambda'})\right) \cong \HH({\M}, {\M} - \bar{p}). $$
Hence, $\mathrm{im} \left( \HH(\mathbb{X}_\alpha, \mathbb{M}^{\alpha, \lambda + 4\delta}) \rightarrow
\HH(\mathbb{X}_{\alpha'}, \mathbb{M}^{\alpha', \lambda'})\right)
\cong \HH({\M}, {\M} - \bar{p})$ by 
Proposition~\ref{persistence-equiv}.
The same argument implies that $\mathrm{im} \left( \HH(\mathbb{X}_\alpha, \mathbb{M}^{\alpha, \lambda + 2\delta}) \rightarrow
(\mathbb{X}_{\alpha'}, \mathbb{M}^{\alpha',\lambda' + 2\delta'})\right)
\cong \HH(\M, \M - \bar{p})$ which establishes the claim in~(\ref{eq10}).
%
%It then follows that $\mathrm{im}(\HH(\mathbb{X}_\alpha, \mathbb{B}^{\alpha,\lambda + 3\delta}) \rightarrow 
%\HH(\mathbb{X}_{\alpha'}, \mathbb{B}^{\alpha',\lambda' + \delta'}) ) \cong  \HH({\M}, {\M} - \bar{p})$ by Proposition~\ref{rank}. 
Eqn (\ref{eqn:im:prop}) then follows from Proposition~\ref{rank}. 
In particular, if $\alpha'=\alpha$,
we have $$\mathrm{im}\left(\HH(\mathbb{X}_\alpha, \mathbb{B}^{\alpha,\lambda + 3\delta}) \rightarrow 
\HH(\mathbb{X}_{\alpha}, \mathbb{B}^{\alpha,\lambda + \delta}) \right) \cong  \HH({\M}, {\M} - \bar{p}).$$
\end{proof}

Finally, we intersect each set with a sufficiently large ball 
$B_r(p)$ so that we only need to inspect within the neighborhood $B_r(p)$ of $p$. 
Specifically, denote $\X_{\alpha, r} = \X_\alpha \cap B_r(p)$ and 
$\X_{\alpha, r}^{\beta} = \X_{\alpha,r} \cap B^{\beta}(p)$. 
We obtain the next proposition by applying the Excision theorem (details in Appendix \ref{appendix:sec:offsets}). 
%\begin{proposition}
%%\label{im:prop:ball}
%Let 
%$\delta=\alpha + 3\eps$ and $\delta'= \alpha' + 3\eps$ where
%$\eps < (3-\sqrt{8})\rho$ 
%and  
%$\max(\eps, \theta_1) \leq \alpha \leq \alpha' \leq 
%\min(\rho(\M)-\eps, \theta_2)$. 
%%$\eps\leq \alpha \leq \alpha' \leq \rho(\M)-\eps$ 
%Let $\eps < \lambda'< \rho(\M)-4\delta'$, $\lambda = \lambda' + 2(\alpha' - \alpha)$, and $r > \lambda + 5\delta$.
%%
%%$\delta=\alpha + 3\eps$ and $\delta'= \alpha' + 3\eps$ where 
%%$\eps\leq \alpha \leq \alpha' \leq \rho(\M)-\eps$ and $\eps < \lambda< \rho(\M)-4\delta'$. 
%Then,
%$\mathrm{im}(\HH(\X_{\alpha, r}, \X_{\alpha, r}^{\lambda+3\delta}) 
%\rightarrow \HH(\X_{\alpha', r}, \X_{\alpha', r}^{\lambda'+\delta'})) \cong \HH({\M}, {\M} - \bar{p})$.
%In particular, 
%$\mathrm{im}(\HH(\X_{\alpha, r}, \X_{\alpha, r}^{\lambda+3\delta}) \rightarrow 
%\HH(\X_{\alpha, r}, \X_{\alpha, r}^{\lambda+\delta})) \cong \HH({\M}, {\M} - \bar{p}).$
%\label{persistence-intersect-ball}
%\end{proposition}
\begin{proposition}
Let all the parameters satisfy the same conditions as in Proposition \ref{im:prop}. 
% and $r > \lambda + 5\delta$.
%
%$\delta=\alpha + 3\eps$ and $\delta'= \alpha' + 3\eps$ where 
%$\eps\leq \alpha \leq \alpha' \leq \rho(\M)-\eps$ and $\eps < \lambda< \rho(\M)-4\delta'$. 
Then, for $r > \lambda + 5\delta$, we have: 
$$\mathrm{im}\left(\HH(\X_{\alpha, r}, \X_{\alpha, r}^{\lambda+3\delta}) 
\rightarrow \HH(\X_{\alpha', r}, \X_{\alpha', r}^{\lambda'+\delta'})\right) \cong \HH({\M}, {\M} - \bar{p}). $$
In particular, 
$\mathrm{im}\left(\HH(\X_{\alpha, r}, \X_{\alpha, r}^{\lambda+3\delta}) \rightarrow 
\HH(\X_{\alpha, r}, \X_{\alpha, r}^{\lambda+\delta})\right) \cong \HH({\M}, {\M} - \bar{p}).$
\label{persistence-intersect-ball}
\end{proposition}

\noindent In fact, one can relax the parameters, and the image homology $\mathrm{im}\left(\HH(\X_{\alpha,r}, \X_{\alpha,r}^{\beta_2})\rightarrow \HH(\X_{\alpha',r}, \X_{\alpha',r}^{\beta_1})\right)$ captures (that is, is isomorphic to) the local homology $\HH(\M, \M-\bar{p})$ as long as $\beta_1 \ge \alpha' + 4\eps$, $\beta_2 \ge \beta_1 + \alpha + \alpha'+6 \eps$ and $r > \beta_2 + 2\alpha + 6\eps$. 

%%%%%%%%%%%%%%%%%%%%%%%%%%%%%
\section{Interleaving Nerves and Rips complexes}
\label{sec:Rips}
%%%%%%%%%%%%%%%%%%%%%%%%%%%%%

We now relate the relative homology of pairs as in Proposition~\ref{persistence-intersect-ball}
to the relative homology of pairs in Rips complexes. Our algorithm
works on these pairs of Rips complexes to derive the local homology
at a point on $\M$. As before, 
let $p\in P$ be a point from the sample.

\myparagraph{Nerves of spaces.} 
Consider the space $\X_{\alpha,r}= \X_\alpha \cap B_r(p)$. 
The connection of such spaces with simplicial complexes (Vietoris-Rips complex in particular) is made through the so-called nerve of a cover. 
In general, let $\cal U$ be a finite collection of sets. The \emph{nerve $\nerv {\cal U}$ of $\cal U$} is a simplicial complex whose simplices are given by all subsets of $\cal U$ whose members have a non-empty common intersection. That is, 
$$\nerv{\cal U} := \{ {\cal A} \subseteq {\cal U} \mid \cap {\cal A} \neq \emptyset \}. $$ 
The set $\cal U$ forms a \emph{good cover} of the union $\bigcup \cal U$ if the intersection of any subsets of $\cal U$ is either empty or contractible. 
The Nerve Lemma states that if $\cal U$ is a good cover of $\bigcup \cal U$, then $\nerv{\cal U}$ is homotopic to $\bigcup \cal U$, denoted by $\nerv{\cal U} \approx \bigcup \cal U$. 

Now consider the set of sets $\clX_{\alpha,r} = \{ B_\alpha(p_i) \cap B_r(p) \mid p_i \in P\}$; note that $\X_{\alpha, r} = \bigcup \clX_{\alpha,r}$. 
Since each set in $\clX_{\alpha,r}$ is convex, $\clX_{\alpha,r}$ forms a good cover of $\X_{\alpha,r}$ and thus $\nerv \clX_{\alpha, r} \approx \X_{\alpha,r}$ by the Nerve Lemma. 
Furthermore, it follows from Lemma A.5 of \cite{bei:2012} that for $r > \beta + 2\alpha$, the set  $\clX_{\alpha,r}^\beta = \{ B_\alpha(p_i) \cap B_r(p) \cap B^\beta(p) \}_{i \in [1, n]}$ also form a good cover of $\bigcup \clX_{\alpha,r}^\beta (= \X_{\alpha,r}^\beta)$; see Appendix \ref{appendix:goodcover} for details. Thus, we have $\nerv\clX_{\alpha,r}^\beta
 \approx \X_{\alpha,r}^\beta$. 
We can now convert the relative homology between $\X_{\alpha,r}$ and $\X_{\alpha,r}^\beta$ to the homology of their nerves. In particular, we have the following result. The proof is in Appendix \ref{appendix:lem:relativenerves}, and it relies heavily on the proof of Lemma 3.4 of \cite{Chazal:2008} which gives a crucial commutative result for the space and its nerve. 

\begin{lemma}
Let 
all the parameters satisfy the same conditions as in Proposition \ref{im:prop}. %\ref{persistence-intersect-ball}. 
Then, for $r > \lambda + 5\delta$: 
\begin{center}
$\mathrm{im}\left(\HH(\nerv\clX_{\alpha, r}, \nerv \clX_{\alpha, r}^{\lambda+3\delta}) \rightarrow 
\HH(\nerv \clX_{\alpha', r}, \nerv \clX_{\alpha', r}^{\lambda'+\delta'}) \right) \cong  \HH({\M}, {\M} - \bar{p}). $
\end{center}
\label{lem:relativenerves}
\end{lemma}

\noindent {\bf Relating nerves and Rips complexes.~}
First, we recall that for $\alpha\geq 0$, the
\emph{\v{C}ech complex} $C^{\alpha}(Q)$ of a point set $Q$ 
is the nerve of the cover $\{B_\alpha(q_i) : q_i\in Q\}$ of 
$\cup B_\alpha(q_i)=\X_\alpha$.
The \emph{Vietoris-Rips} (Rips in short) complex $\Rips^{\alpha}(Q)$ is the 
maximal complex induced by the edge set $\{(p_j, p_k) \mid d(p_j, p_k) \le \alpha \}$. 
It is well known that for any point set $Q$, the following holds: 
$$
C^{\alpha}(Q) \subset \Rips^{2\alpha}(Q) \subset C^{2\alpha}(Q).
$$
%\begin{wrapfigure}{l}{2.2in}
\parpic[r]{\includegraphics[height=4cm]{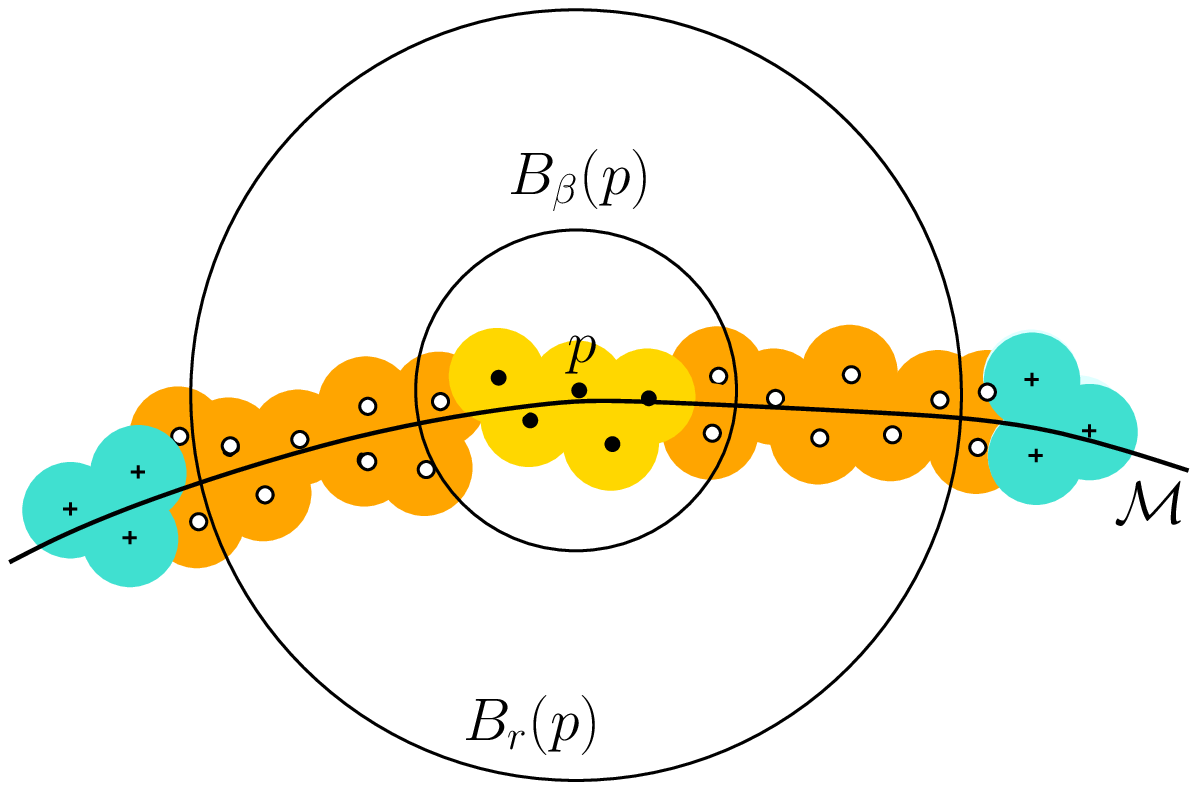}}
%\includegraphics[height=1.57in]{Palpha-r}
%\end{wrapfigure}
 Define $P_{\alpha, r} = \{p_i \in P \mid B_\alpha(p_i) \cap B_r(p) \neq \emptyset  \}$. 
Obviously, $P_{\alpha,r}$ forms the vertex set 
for the nerve $\nerv \clX_{\alpha, r}$. 
Similarly, let $P_{\alpha, r}^{\beta} = \{p_i\in P_{\alpha,r} \mid
B_\alpha(p_i) \cap B^{\beta}(p) \neq \emptyset\} $ denote the
vertex set of $\nerv\clX_{\alpha, r}^{\beta}$. 
See the figure on right for an example, where the union of solid and empty dots forms the set of points $P_{\alpha,r}$, while $P_{\alpha,r}^\beta$ consists the set of empty dots. 
Note that from the definition, it follows 
that $P_{\alpha, r}^{\beta} \subset P_{\alpha, r}$ and 
$P_{\alpha, r}^{\beta} \subset P_{\alpha, r}^{\beta'}$ for $\beta' < \beta$.
Furthermore, as the offset $\X_\alpha$ grows, it is immediate that 
$P_{\alpha, r} \subset P_{\alpha', r}$
and $P_{\alpha, r}^{\beta} \subset P_{\alpha', r}^{\beta}$ 
for $\alpha < \alpha'$. 

Each element in the good cover 
$\clX_{\alpha,r}$ or $\clX_{\alpha,r}^{\beta}$ is in 
the form of $B_\alpha(p_i) \cap B_r(p)$ or $B_\alpha(p_i) \cap B_r(p) \cap B^{\beta}(p)$.
Since the \v{C}ech complex of a set is the nerve of the 
set of balls $B_\alpha(p_i)$,
it follows easily that 
\begin{eqnarray}
\nerv\clX_{\alpha, r} \subset C^{\alpha}(P_{\alpha, r}) \subset \Rips^{2\alpha}(P_{\alpha, r})
\mbox{ and } 
\nerv\clX_{\alpha, r}^{\beta} \subset  C^{\alpha}(P_{\alpha, r}^{\beta}) \subset \Rips^{2\alpha}(P_{\alpha, r}^{\beta}).
\label{eqn:NerveinRips}
\end{eqnarray}

\begin{claim}
(i) $\Rips^{2\alpha}(P_{\alpha, r}) \subset \nerv\clX_{3\alpha, r}$, and 
(ii) $\Rips^{2\alpha}(P_{\alpha, r}^{\beta}) \subset \nerv\clX_{3\alpha, r}^{\beta}$.
\label{claim:RipsinNerve}
\end{claim}

\begin{proof}
To prove (i), consider an arbitrary simplex 
$\sigma=[p_{0}p_{1}\ldots p_{\ell}] \in \Rips^{2\alpha}(P_{\alpha, r})$. 
By definition of Rips complex, $d(p_i, p_j) \le 2\alpha$ for $0\leq i, j \leq \ell$. 
Then, for any point $x \in B_\alpha(p_0) \cap B_r(p)$, we have that 
$d(x,p_{i}) < d(x,p_0) + d(p_{0},p_{i}) < 3\alpha$ implying
$x\in \cap_{i=0}^{\ell} B_{3\alpha}(p_i)$ and $(\cap_{i=0}^{\ell} B_{3\alpha}(p_i)) \cap B_r(p) \neq \emptyset$. In other words, $\sigma \in \nerv\clX_{3\alpha, r}$, thus proving Claim (i). 
Claim (ii) can be shown by a similar argument. 
\end{proof}
%To prove Claim (ii), similarly, take any simplex $\sigma=[p_{0}p_{1}\ldots p_{\ell}] \in \Rips^{2\alpha}(P_{\alpha, r}^{\beta})$. 
%\begin{enumerate}
%\item $B_\alpha(p_i) \cap B_r(p) \cap B^{\beta}(p) \neq \emptyset$ for $0\leq i \leq \ell$; 
%\item $B_\alpha(p_i) \cap B_\alpha(p_j)\neq \emptyset$ for $0\leq i, j \leq \ell$ ;
%\item $(\cap_{i=0}^{\ell} B_\alpha(p_i)) \cap B_r(p) \cap B^{\beta}(p) = \emptyset$.
%\end{enumerate}
%
%Let $x \in B_\alpha(p_0) \cap B_r(p) \cap B^{\beta}(p)$. Then 
%$d(x,p_{i}) < d(x,p_0) + d(p_{0},p_{i}) < 3\alpha$ as 
%shown in Fig.\ref{fig:mini:subfig} (\tamal{Figure needs to change}).
%Therefore, $x\in \cap_{i=0}^{\ell} B_{3\alpha}(p_i)$ and 
%$(\cap_{i=0}^{\ell} B_{3\alpha}(p_i)) \cap B_r(p) \cap B^{\beta}(p) 
%\neq \emptyset$. 
%In other words, $\sigma \in \mathrm{Nr}(\X_{3\alpha, r}^{\beta})$, that is, 
%\begin{eqnarray}
%\Rips^{2\alpha}(P_{\alpha, r}^{\beta}) \subset \mathrm{Nr}(\X_{3\alpha, r}^{\beta})
%\label{eq7}
%\end{eqnarray}

Set $\eta_1 = \lambda + 9\alpha + 3\eps$ and $\eta_2 \ge \eta_1 + 12\alpha+6\eps$ for any $\lambda > \eps$.  %$\eta_2 = \lambda + 21\alpha + 9\eps$ where $\lambda > \eps$. 
Combining Eqn (\ref{eqn:NerveinRips}) and Claim \ref{claim:RipsinNerve}, 
we get three nested sequences
 $$
 \nerv\clX_{\alpha, r} 
 \subset \Rips^{2\alpha}(P_{\alpha, r})
 \subset \nerv\clX_{3\alpha, r} 
 \subset \Rips^{6\alpha}(P_{3\alpha, r})
 \subset \nerv\clX_{9\alpha, r} 
 $$ 
  $$
 \nerv\clX_{\alpha, r}^{\eta_1} 
  \subset \Rips^{2\alpha}(P_{\alpha, r}^{\eta_1})
 \subset \nerv\clX_{3\alpha, r}^{\eta_1} 
  \subset \Rips^{6\alpha}(P_{3\alpha, r}^{\eta_1})
 \subset \nerv\clX_{9\alpha, r}^{\eta_1} 
 $$
  $$
 \nerv\clX_{\alpha, r}^{\eta_2} 
 \subset \Rips^{2\alpha}(P_{\alpha, r}^{\eta_2})
 \subset \nerv\clX_{3\alpha, r}^{\eta_2} 
 \subset \Rips^{6\alpha}(P_{3\alpha, r}^{\eta_2})
 \subset \nerv\clX_{9\alpha, r}^{\eta_2} 
 $$ 
These give rise to the following sequence of pairs 
\[
(K_\alpha, K_\alpha^{\eta_2}) \hookrightarrow 
(R_{\alpha}, R_{\alpha}^{\eta_2}) \hookrightarrow 
(K_{3\alpha}, K_{3\alpha}^{\eta_2}) \hookrightarrow 
(K_{3\alpha}, K_{3\alpha}^{\eta_1}) \hookrightarrow 
(R_{3\alpha}, R_{3\alpha}^{\eta_1}) \hookrightarrow 
(K_{9\alpha}, K_{9\alpha}^{\eta_1})
\]
where $K_\alpha = \nerv\clX_{\alpha, r}$, 
$K_\alpha^{\beta}=\nerv\clX_{\alpha, r}^{\beta}$, 
$R_{\alpha}=\Rips^{2\alpha}(P_{\alpha, r})$ and $R_{\alpha}^{\beta}=\Rips^{2\alpha}(P_{\alpha, r}^{\beta}) $.
%Its corresponding sequence at the homology level is
%\[
%\HH(K_\alpha, K_\alpha^{\eta_2}) \rightarrow 
%\HH(R_{\alpha}, R_{\alpha}^{\eta_2}) \rightarrow 
%\HH(K_{3\alpha}, K_{3\alpha}^{\eta_2}) \rightarrow 
%\HH(K_{3\alpha}, K_{3\alpha}^{\eta_1}) \rightarrow 
%\HH(R_{3\alpha}, R_{3\alpha}^{\eta_1}) \rightarrow 
%\HH(K_{9\alpha}, K_{9\alpha}^{\eta_1}) .
%\]
From Proposition~\ref{persistence-intersect-ball} and 
Lemma \ref{lem:relativenerves}, it is immediate that 
$\mathrm{im}({i_\alpha}_{*}) \cong  \mathrm{im}({i_{3\alpha}}_{*}) 
\cong \HH({\M}, {\M} - \bar{p})$ where
${i_\alpha}_{*}$ and ${i_{3\alpha}}_{*}$ are induced from 
$i_\alpha : (K_\alpha, K_\alpha^{\eta_2}) \hookrightarrow (K_{9\alpha}, K_{9\alpha}^{\eta_1})$ 
and 
$i_{3\alpha} : (K_{3\alpha}, K_{3\alpha}^{\eta_2}) \hookrightarrow (K_{3\alpha}, K_{3\alpha}^{\eta_1})$ .
It follows from Proposition~\ref{rank}
that $\mathrm{im}({j_{\alpha}}_{*}) 
\cong \HH({\M}, {\M} - \bar{p})$ where ${j_{\alpha}}_{*}$ is induced from 
$j_{\alpha} : (R_{\alpha}, R_{\alpha}^{\eta_2}) 
\hookrightarrow (R_{3\alpha}, R_{3\alpha}^{\eta_1})$.
To apply Proposition~\ref{persistence-intersect-ball}, we need the condition required by 
Eq. \ref{eq10}, which is $\eta_2 + \alpha + 3\eps < \rho(\M)$ here. 
This condition
together with $\eta_2 \ge \eta_1 + 12\alpha+6\eps$ require 
that $\alpha < \frac{\rho(\M)-13\eps}{22}$.
We also need $\theta_1 \leq \alpha$. 
Both conditions are satisfied when $0<\eps<\frac{\rho(\M)}{58}$.
Thus, we have our main result:
\begin{theorem}
Let $0<\eps < \frac{\rho(\M)}{58}$ 
and  $\theta_1 \leq \alpha \leq \frac{\rho(\M)-13\eps}{22}$. 
Furthermore, let $\eta_1$ and
$\eta_2$ be such that  $\eps<\eta_1, \eta_2<\rho(\M)$,
$\eta_1 \ge 9\alpha + 4\eps$, and $\eta_2 \ge \eta_1 + 12\alpha+6\eps$. 
The inclusion 
$$j_{\alpha} : (\Rips^{2\alpha}(P_{\alpha, r}), \Rips^{2\alpha}(P_{\alpha, r}^{\eta_2})) 
\hookrightarrow 
(\Rips^{6\alpha}(P_{3\alpha, r}), \Rips^{6\alpha}(P_{3\alpha, r}^{\eta_1}))$$
satisfies $\mathrm{im}({j_{\alpha}}_{*}) 
\cong \HH({\M}, {\M} - \bar{p})$ for any $r \ge \eta_1+\eta_2$. 
\label{thm:mainthm}
\end{theorem}

\paragraph{Algorithm.} 
Given a sample point $p = p_i$, our algorithm first constructs the 
necessary Rips complexes as specified in Theorem \ref{thm:mainthm} for some parameters $\alpha < \eta_1 < \eta_2 < r$.
For simplicity, rewrite $j_\alpha: (A_1, B_1) \hookrightarrow (A_2,B_2)$
where $B_1 \subset A_1 \subset A_2$ and $B_1 \subset B_2 \subset A_2$.
After obtaining the necessary Rips complexes, 
one possible method for computing $\mathrm{im}(j_{\alpha*})$ would be to
cone the subcomplexes $B_1$ and $B_2$ with a dummy vertex 
$w$ to obtain an inclusion 
$\iota : A_1 \cup (w*B_1) \hookrightarrow A_2 \cup (w*B_2)$
where $w*B_j = B_j \cup \{w*\sigma | \sigma \in B_j\}$ 
is the cone on $B_j$ $(j=1,2)$.
It is easy to see that $\mathrm{im}(j_{\alpha*}) \cong \mathrm{im}(\iota_*)$. 
%change the relative homology group $\HH(A_1, B_1)$  
%to homology group $\HH(A_1 \cup \{w * B_1\})$ where 
%$w*B_1 = B_1 \cup \{w*\sigma | \sigma \in B_1\}$ is the cone on $B_1$
%and similarly for $\HH(A_2, B_2)$. 
Then, the standard persistent homology algorithm can be applied.
However, the cone operations may add many unnecessary simplices
slowing down the computation. 
Instead, we order the simplices in $A_2$ properly to 
build a filtration so that  
the rank of $\mathrm{im}({j_\alpha}_*)$
can be read off from the reduced boundary matrix built from the filtration. 
The details of this algorithm can be found in 
Appendix \ref{appendix:matrix-algorithm}.

\section{Experimental results}
\label{exp}
We present some preliminary experimental results on several synthesized and 
real data. Recall that our method only needs points 
in the neighborhood of a base point. 
While the theoretical result guarantees
the correct detection of dimension for correct choices of parameters, 
in practice, the choice of the base point plays an important role.
If the points sample only a 
patch of a manifold, then the local homology of points near the boundary of 
that patch will be trivial, which results in 
plenty of base points with trivial local homology.
Furthermore, noise and 
inadequate density make the dimension estimation difficult.
To overcome these hurdles, we explore some practical strategies.
\begin{figure}[h!]
\begin{center}
\begin{tabular}{ccc}
\includegraphics[scale=0.8]{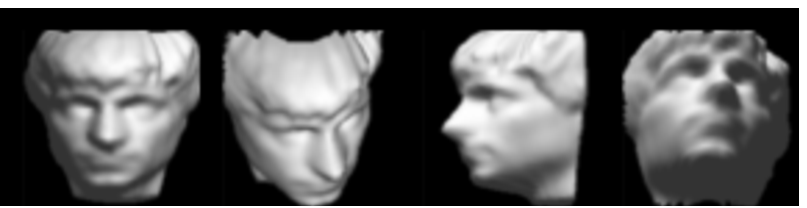} &
\includegraphics[scale=0.8]{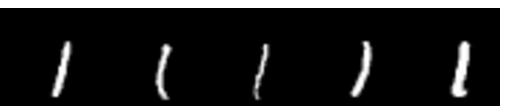} &
\includegraphics[scale=0.8]{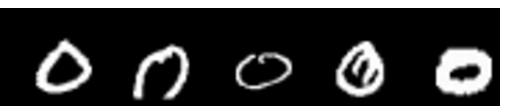}\\ 
(a) \textbf{Head} &
(b) \textbf{D1} &
(c) \textbf{D0} 
\end{tabular}
\end{center}
\caption{Image data : rotating head (\textbf{Head}), handwritten ones (\textbf{D1}) and zeros (\textbf{D0}).}% from the MNIST database.}
\label{fig:real-data}
\end{figure}

For the synthesized data, which is uniform and dense, 
we take a sparse and uniform subsample from the input
as a set of base points. At each base point, the local 
homology is estimated by our program. We discard the result in which 
the computed homology is trivial or does not coincide with 
$\redHH(\sphere^n)$ for any $n$, as these are obviously not correct. 
%This can happen either because of noise, inadequate density, or
%the presence of boundaries.
The remaining base points return the homology 
of an $n$-sphere, that is $\rank(\HH_i)=1$ \emph{iff} $i=n$ for
some $n$. These are called \emph{valid base points}. 
These points are grouped according to which $n$-sphere homology
they have, and we return the dimension $n$ of 
the group with most members as
the detected dimension. 
% intrinsic dimension of the manifold.
\begin{table}[ht!]
\centering
{\small
\begin{tabular}{|c|c|c|c|c|c|c|c|}
\hline
& {\sc Sample Points} & {\sc Avg. Neighb.} & {\sc Not} $n$-{\sc sphere} & 
{\sc Trivial} & \multicolumn{2}{|c|}{$n$-{\sc sphere}} & {\sc Correct Ratio} \\
\hline
$\sphere^3$ &  $4096$ &  $19$ & $0/60$ & $38/60$ &  n=3 & $22/60$ & $100\% (22 / 22)$ \\
\hline
$\sphere^4$ & $4097$ &  $34$ & $0/46$  & $40/46$ & n=4 & $6/46$  & $100\% (6 / 6)$ \\
\hline
$\sphere^5$ & $32769$ &  $52$ & $0/74$ & $69/74$ & n=5 & $5/74$ & $100\% (5 / 5)$\\
\hline
$\sphere^{6}$ & $262145$ &  $74$ &$0/220$ & $213/220$ & n=6 & $7/220$ & $100\% (7 / 7)$\\
\hline
\textbf{Shift} & $2240$ &  $37$ & $0/67$ & $15/67$ &  n=2 & $52/67$ & $100\% (52 / 52)$ \\
\hline
\multirow{2}{*}{${ M}^{3}$} & 
\multirow{2}{*}{$2796$ } &
\multirow{2}{*}{$316$ } &
\multirow{2}{*}{$0/54$} &
\multirow{2}{*}{$40/54$} &
n=2 & $1/54$  & \multirow{2}{*}{$92.8\% (13 / 14)$}  \\
& & & & & n=3& $13/54$ &\\
\hline
\end{tabular}
}
\caption{Results for synthetic data}
\label{table:synthetic}
\end{table}

For the real data, which mostly comes from a small part of a manifold,
we use a different strategy because these data are
non-uniform and contain high noise and outliers. 
Three data \textbf{Head}, \textbf{D1}, and \textbf{D0} (some
samples shown in Figure~\ref{fig:real-data}) are considered. 
We first identify some sample points called centers 
away from the boundary and undersampled regions
using a graph based method described in the Appendix~\ref{appendix:centers}.
Then, we estimate the local homology at these points. Table~\ref{real-tab} 
in the Appendix~\ref{appendix:centers} provides the results on estimated
dimensions.

Our synthetic data consists of points sampled from spherical caps of $n$-spheres $\sphere^n$ for $n = 3, 4, 5, 6$; 
a $3$-manifold $M^3\subset \reals^{50}$ with boundary (computed from a parametric equation); 
and a $2$D translation of a smaller image 
within a black image with resolution 
$60 \times 84$(\textbf{Shift}) (see \cite{CC09}). 
The input for each $\sphere^n$ is a uniform $0.0125$-sample of
a spherical cap (thus is a manifold \emph{with} boundary) with no noise. 
The \textbf{Shift} data is also noiseless.
%The 3-manifold $M^3$ has the following parametric equation:
%$$(u, u + v, u + v + w, (u + v + w)(u + v), . . . , (u + v + w)(u + v)^{47} ).$$
%{\sc Sample Points} column shows the number of sample points from
%the manifold, {\sc Used Points} column shows
%the approximate number of points in the local
%neighborhood of the base point our algorithm used to estimate local homology, 
%$n$-{\sc Sphere} column shows the number of base points with the
%homology of an $n$-sphere for each $n$,
%{\sc Trivial} column shows the number of base points with 
%trivial homology, {\sc Not} $n$-{\sc Sphere} column shows the 
%number of points with a homology not that of an $n$-sphere and not trivial
%and {\sc Correct Ratio} column shows the rate of correct dimension 
%detections over all valid base points. The columns of n-sphere, Not n-sphere 
%and Trivial also contain the
%total number of base points for comparison. \\
%
%\begin{wrapfigure}{l}{3in}
%\begin{tabular}{|l|c|c|c|c|}
%\hline
%& \textbf{Shift} & \textbf{Head} & \textbf{D1} & \textbf{D0}\\
%\hline
%Ours & $2$ & $3$ & $4$ & $3$\\
%SLIVER &  $3$ & $4$ & $3$ & $2$\\
%MLE & $4.27$ & $4.31$ & $11.47$ & $14.86$ \\
%MA & $3.35$ & $4.47$ & $10.77$ & $13.93$ \\
%PN & $3.62$ & $3.98$ & $6.22$ & $8.86$\\
%LPCA & $3$ & $3$ & $5$ & $8.86$\\
%ISOMAP & $2$  & $3$ & $5$ &$[3,6]$\\
%\hline
%\end{tabular}
%%\caption{Comparison results. The $2$nd-$7$th rows are from \cite{CC09}}
%\label{table:comparison}
%%\end{table}
%\end{wrapfigure}
The sample points of $M^3$ is noisy with a 
$0.05$ unit Hausdorff noise. 
The results on the synthetic data are summarized in Table~\ref{table:synthetic}.
%Table \ref{table:sphere}. 
{\sc Avg. Neighb.} column gives the average number of points in the local neighborhood of each base point used to estimate local homology. 
{\sc Correct Ratio} column shows the ratio of correct dimension 
detection over all valid base points. 
Among all valid base points, our algorithm 
produces no false positives for all the $\sphere^n$ data sets.
For the noisy sample of $M^3$, we 
have only one false positive out of 14 valid points. 
The high number of points that return trivial homology (5th column) is mainly due to points near the boundary of the manifold. 
For the \textbf{Shift} data, our method detects its dimension $2$
with high confidence. 
The \textbf{Shift} was used and compared in \cite{CC09}.

%Table \ref{table:comparison}. 
\begin{wrapfigure}{l}{3in}
\begin{tabular}{|l|c|c|c|c|}
\hline
& \textbf{Shift} & \textbf{Head} & \textbf{D1} & \textbf{D0}\\
\hline
Ours & $2$ & $3$ & $4$ & $3$\\
SLIVER &  $3$ & $4$ & $3$ & $2$\\
MLE & $4.27$ & $4.31$ & $11.47$ & $14.86$ \\
MA & $3.35$ & $4.47$ & $10.77$ & $13.93$ \\
PN & $3.62$ & $3.98$ & $6.22$ & $8.86$\\
LPCA & $3$ & $3$ & $5$ & $8.86$\\
ISOMAP & $2$  & $3$ & $5$ &$[3,6]$\\
\hline
\end{tabular}
%\caption{Comparison results}
\label{table:comparison}
%\end{table}
\end{wrapfigure}
In the table on left, we show comparisons with other methods.
Although \textbf{Shift} is uniform and noise free, 
only ISOMAP and ours get the correct dimension.
The real data contains $698$ images of a rotating 
head (\textbf{Head}, Fig. \ref{fig:real-data}(a)), 
$6742$ images of handwritten ones 
(\textbf{D1}, Fig. \ref{fig:real-data}(b)) and 
$5923$ images of handwritten zeros 
(\textbf{D0}, Fig. \ref{fig:real-data}(c)) from MNIST database. 
These three data were also explored and compared in \cite{CC09}, where  
Cheng and Chiu \cite{CC09} compared their dimension detection 
method via sliver (SLIVER) with  
other methods: the maximum likelihood estimation  (MLE) \cite{Elizaveta2005},
the manifold adaptive method (MA) \cite{FSA07},
the packing number method (PN) \cite{K02}, 
the local PCA (LPCA) \cite{CWW05},
and the isomap method (ISOMAP) \cite{Tenenbaum:2000}.
Since we test our method on the same data, 
we include the comparison results on these three data 
along with \textbf{Shift} data 
from \cite{CC09} in the table 
%\ref{table:comparison}, 
where all rows except the first row are from \cite{CC09}.
Details and statistics of our experiments on real data are presented in
the Appendix~\ref{appendix:centers}.

%We apply our second strategy on \textbf{Head}, \textbf{D1} and \textbf{D0}.
%All of them have only one major component in the graph which 
%connects two points within the distance of several times of the 
%closest pairwise distance of the sample points. 
%There are $696$ points  in the major component of \textbf{Head},
%$4262$ points in that of \textbf{D1} and  
%$5914$ points in that of \textbf{D0}.
%For \textbf{Head}, a subsample of $138$ points is taken 
%from $505$ points in the neighborhood of the center of the major component.
%The computation of the local homology at the center by using that subsamples
%reports dimension $3$ for \textbf{Head}. 
%The dimension of \textbf{Head} is considered to be around $3$ or $4$ in literature.
%Ours falls into this range.
%We subsample $183$ points from $2069$ points in the neighborhood of 
%the center of \textbf{D1}, and $102$ points from $3494$ points around \textbf{D0}'s center.
%Our local homology computation on these two subsamples reports dimension $4$ for both 
%\textbf{D1} and \textbf{D0}. Although the ground-truth dimension for \textbf{D1} and 
%\textbf{D0} are unknown, ours along with SLIVER, PN, LPCA and ISOMAP report
%dimension in range $[3,9]$ for  \textbf{D1} and range $[2,9]$ for \textbf{D0}.
%%
%%

\section{Conclusions}
In this paper, we present a topological method to estimate the dimension of a manifold from its point samples with a theoretical guarantee.
%In recent years, topological methods have been shown to be useful for more and more data analysis applications, and this paper provides another result in this direction. 
The use of local topological structures helps to alleviate the dependency of our method on the regularity of point samples, and the use of persistent homology for a pair of homology groups (instead of a single homology group) helps to increase its robustness. 

It will be interesting to investigate other data analysis problems 
where topological methods, especially those based on 
local topological information (yields to efficient computations), 
may be useful.
Currently, we have conducted some preliminary experiments to demonstrate the performance of our algorithm. It will be interesting to conduct large-scale experiments under a broad range of practical scenarios, so as to better understand data in those contexts. 

\vfill\eject
\newpage
\bibliographystyle{abbrv}
\bibliography{mybib}

\newpage
\appendix

\section{Proof for Proposition \ref{ball-ball-lem}}
\label{appendix:ball-ball-lem}

We only need to show that $i_*'$ is an isomorphism
as Proposition~\ref{point-ball-lem} proves it for $i_*$.
Since the inclusion induced homomorphisms
$j_*: \HH(\M,\M-\cirD_2)\rightarrow \HH(\M,\M-z)$ 
and $i_*: \HH(\M,\M-\cirD_1)\rightarrow \HH(\M,\M-z)$ are isomorphisms
by Proposition~\ref{point-ball-lem} and $j_*=i_*\circ i'_*$,
we have that $i'_*$ is an isomorphism as well. 

\section{Missing Details in Section \ref{sec:offsets}}
\label{appendix:sec:offsets}

\paragraph{Proof of Proposition \ref{projpair}.}
By Proposition~\ref{deform}, the map $\pi_{\alpha*}$ is an isomorphism.
By Proposition~\ref{ball-intersect}, 
$D_\beta$ is a closed topological ball as $\beta<\rho(\M)$. 
Hence, $(\M,\M-\cirD_\beta)\hookrightarrow (\M,\M-\bar{p})$ induces
the isomorphism $i_*$ at the homology level, 
see Proposition~\ref{point-ball-lem}.
The observation then follows. 

\paragraph{Missing details for interleaving in section~\ref{sec:offsets}.}
From Eq.~\ref{dist-nest}, it follows that for any 
$\lambda\in(\eps, \rho(\M)-\delta)$:
\begin{eqnarray}
\MM_{\alpha,\lambda}=\pi^{-1}_\alpha(\mathring{B}_{\lambda}(p) \cap {\M}) \subset 
\pi^{-1}_\alpha(\mathring{B}_{\lambda+\eps}(\bar{p}) \cap {\M}) \subset
\mathring{B}_{\lambda+\delta - \eps}(\bar{p}) \cap \X_\alpha \subset
\mathring{B}_{\lambda+\delta}(p) \cap \X_\alpha=\BB_{\alpha,\lambda+\delta} .
\label{eq1}
\end{eqnarray}
Now take a point $x \in B_\lambda(p) \cap \X_\alpha$. 
Then $d(\pi(x) , p) \leq d(\pi(x) ,x) + d(x,p) \leq (\alpha+\eps) + \lambda$. 
Therefore, 
\begin{eqnarray}
\BB_{\alpha,\lambda}=B_\lambda(p) \cap \X_\alpha \subset 
\pi^{-1}_{\alpha}(B_{\lambda+\delta-2\eps}(p) \cap {\M})=\MM_{\alpha, \lambda+\delta-2\eps} \subset
\MM_{\alpha,\lambda+\delta}.
\label{eq2}
\end{eqnarray}
Eq.~\ref{eq1} and Eq.~\ref{eq2} provide the required nesting:
$$
\MM_{\alpha,\lambda} \subseteq \BB_{\alpha,\lambda+\delta}
\subseteq \MM_{\alpha,\lambda+2\delta}.
$$
\paragraph{Proof of Proposition \ref{persistence-intersect-ball}.}
%\begin{proposition}
%\label{im:prop:ball}
%For $\eps\leq\alpha\leq \rho(\M)-\eps$ and $\lambda < \rho(\M)-4\delta$,
%$
%\mathrm{im}(\HH(\X_{\alpha, r}, \X_{\alpha, r}^{\lambda+3\delta}) \rightarrow \HH(\X_{\alpha, r}, \X_{\alpha, r}^{\lambda+\delta})) \cong \HH({\M}, {\M} - z)
%$
%\end{proposition}
Recall that by definition ${\BB}^{\alpha,r} = \X_\alpha - \mathring{B}_r(p)$. 
Then, for sufficient large $r > \beta + \alpha + 3\eps$, the closure of ${\mathrm{int}\,}{\BB}^{\alpha,r}$ 
is a subset of $\mathrm{int}\,({\MM}^{\alpha,\beta})$ or 
$\mathrm{int}\,({\BB}^{\alpha,\beta})$.
By the excision theorem, it follows that 
$$\HH(\X_\alpha, {\MM}^{\alpha,\beta}) 
\cong \HH(\X_\alpha - {\mathrm int}\,{\BB}^{\alpha,r}, 
{\MM}^{\alpha,\beta} - {\mathrm int}\,{\BB}^{\alpha,r})$$ 
and 
$$\HH(\X_\alpha, {\BB}^{\alpha,\beta}) 
\cong \HH(\X_\alpha - {\mathrm int}\,{\BB}^{\alpha,r}, {\BB}^{\alpha,\beta} - 
{\mathrm int}\,{\BB}^{\alpha,r}) = \HH(\X_{\alpha, r}, \X_{\alpha, r}^{\beta}),$$ 
where the isomorphisms are 
induced from canonical inclusions. The nested sequence of pairs involves only inclusion maps.
If we repeat the arguments for Proposition~\ref{im:prop} for sets 
intersecting the ball $B_r(z)$ and use Persistence 
Equivalence Theorem~\cite{EH09}, 
we get the claim of this proposition. 
To make sure that $r$ is large enough, we need that $r > \lambda + 4\delta + \alpha + 3\eps$, as well as $r > \lambda' + 2\delta + \alpha' + 3\eps$. We choose $r > \lambda + 5\delta$ to guarantee that. 

\section{Missing Details in Section \ref{sec:Rips}}
\label{appendix:sec:Rips}

\subsection{Good Cover}
\label{appendix:goodcover}

Here we prove that the set of sets $\clX_{\alpha,r}^\beta := \{ B_\alpha(p_i) \cap B_{r}(p) \cap B^\beta(p) \mid p_i \in P \}$ is a good cover for $\bigcup \clX_{\alpha,r}^\beta = \X_{\alpha,r}^\beta$. 

For convenience, denote $F_j = B_\alpha(p_i) \cap B_{r}(p) \cap B^\beta(p)$. 
Note that since $r > \beta + 2\alpha$, we have that any ball $B_\alpha(p_i)$ may intersect the boundary $\partial B_r(p)$ of $B_r(p)$, or the boundary $\partial B^\beta(p)$ of $B_\beta(p)$, (or none of the two boundaries,) but not both. 
In other words, the set $F_j$ can be of three types: 
(i) a complete ball $B_\alpha(p_i)$; (ii) a convex set which is the intersection between $B_\alpha(p_i)$ and $B_r(p)$, but not intersecting the boundary $\partial B_\beta(p)$; and (iii) a potentially non-convex set which is the difference $B_\alpha(p_i) - \mathring{B}_\beta(p)$, but not intersecting the boundary $\partial B_r(p)$. 

Now consider any subset of $\clX_{\alpha,r}^\beta$ with non-empty intersection: Since $B_\alpha(p_i)$ cannot intersect $\partial B_r(p)$ and $\partial B_\beta(p)$ simultaneously, such a subset either only consists of balls from type (i) and (ii), or from type (i) and (iii). 
Since type (i) and (ii) are both convex, their intersection must be contractible. 
If the subset consists of type (i) and (iii), then the result from Lemma 6.7 of \cite{bei:2012} shows that it is also contractible. 
Hence, the intersection of any subset of $\clX_{\alpha,r}^\beta$ is contractible, and as such $\clX_{\alpha,r}^\beta$ forms a good cover for $\X_{\alpha,r}^\beta$. By Nerve Lemma, this implies that $\nerv\clX_{\alpha,r}^\beta$ is homotopic to $\X_{\alpha,r}^\beta$; that is, 
$\nerv\clX_{\alpha,r}^\beta \approx \X_{\alpha,r}^\beta$. 

\subsection{Proof of Lemma \ref{lem:relativenerves}}
\label{appendix:lem:relativenerves}

First, we quote the following result shown in \cite{Chazal:2008}, which states that the isomorphism induced by the homotopy equivalence between a nerve and its space commute with the canonical inclusions on the spaces at the homology level. 
To be consistent with the notations of \cite{Chazal:2008}, let ${\cal N}{\cal U}$ denote the nerve on a good cover ${\cal U}$. 
\begin{proposition}[Lemma 3.4 in \cite{Chazal:2008}]
Let $X \subset X'$ be two paracompact spaces, and Let 
${\cal U}=\{U_i\}_{i\in J}$ and ${\cal U'}=\{U'_i\}_{i \in J}$ be two 
good open covers of
$X$ and $X'$ respectively, based on a same finite parameter set $J$, such that $U_i \subset U'_i$ for all $i \in J$. Then, there exist homotopy 
equivalences ${\cal N}{\cal U} \rightarrow X$ and ${\cal N}{\cal U}' \rightarrow X'$ which commute with the canonical inclusions $X\hookrightarrow X'$
and ${\cal N}{\cal U} \hookrightarrow {\cal N}{\cal U}'$ at homology and homotopy levels.
\end{proposition}

Extending the arguments in the proof of this lemma, we have the following relative homology version. We first give the proof of this result here, after which we explain how Lemma \ref{lem:relativenerves} follows from this result. 

\begin{proposition}
Let Let $X \subset X' \subset X''$ be two paracompact spaces, and Let ${\cal U}=\{U_i\}_{i \in J}$, ${\cal U'}=\{U'_i\}_{i \in J}$ 
and ${\cal U''}=\{U''_i\}_{i \in J}$ be three good open covers of
$X$, $X'$ and $X''$ respectively, based on a same finite parameter set $J$, such that $U_i \subset U'_i \subset U''_{i}$ for all $i \in J$. 
There exist commutative diagrams,
\[
\begin{tabular}{cc}
\xymatrix{
\HH(X', X) \ar[d] \ar[r] & \HH(X'', X)\ar[d] \\
\HH({\cal N}{\cal U}', {\cal N}{\cal U}) \ar[r] & \HH({\cal N}{\cal U}'',{\cal N}{\cal U})
}
&
\xymatrix{
\HH(X'', X) \ar[d] \ar[r] & \HH(X'', X')\ar[d] \\
\HH({\cal N}{\cal U}'', {\cal N}{\cal U}) \ar[r] & \HH({\cal N}{\cal U}'',{\cal N}{\cal U}')
}
\end{tabular}
\]
 where horizontal maps are induced from canonical inclusions $(X', X) \hookrightarrow (X'',X)$, $({\cal N}{\cal U}', {\cal N}{\cal U})\hookrightarrow ({\cal N}{\cal U}'',{\cal N}{\cal U})$, $(X'', X) \hookrightarrow (X'',X')$, $({\cal N}{\cal U}''$, and ${\cal N}{\cal U})\hookrightarrow ({\cal N}{\cal U}'',{\cal N}{\cal U}')$; while vertical maps are isomorphisms.
\label{nerve-relative}
\end{proposition}
%\section{Proof of Relative Homology of Nerves}
%\label{appendix:nerves}
%\begin{proposition}
%Let Let $X \subset X' \subset X''$ be two paracompact spaces, 
%and Let ${\cal U}=\{U_i\}_{i \in J}$, ${\cal U'}=\{U'_i\}_{i \in J}$ 
%and ${\cal U''}=\{U''_i\}_{i \in J}$ be three good open covers of
%$X$, $X'$ and $X''$ respectively, based on a same finite 
%parameter set $J$, such that $U_i \subset U'_i \subset U''_i$ 
%for all $i \in J$. 
%There exist commutative diagrams,
%\[
%\begin{tabular}{cc}
%
%\xymatrix{
%\HH(X'', X) \ar[d] \ar[r] & \HH(X'', X')\ar[d] \\
%\HH({\cal N}{\cal U}'', {\cal N}{\cal U}) \ar[r] & \HH({\cal N}{\cal U}'',{\cal N}{\cal U}')
%}
%&
%
%\xymatrix{
%\HH(X'', X) \ar[d] \ar[r] & \HH(X', X)\ar[d] \\
%\HH({\cal N}{\cal U}'', {\cal N}{\cal U}) \ar[r] & \HH({\cal N}{\cal U}',{\cal N}{\cal U})
%}
%
%\end{tabular}
%\]
% where horizontal maps are induced from canonical inclusions $(X'', X) \hookrightarrow (X'',X')$, $({\cal N}{\cal U}'', {\cal N}{\cal U})\hookrightarrow ({\cal N}{\cal U}'',{\cal N}{\cal U}')$,
% $(X'', X) \hookrightarrow (X',X)$, $({\cal N}{\cal U}'', {\cal N}{\cal U})\hookrightarrow ({\cal N}{\cal U}',{\cal N}{\cal U})$
% and vertical maps are isomorphisms.
%\end{proposition}
\begin{proof}
From the good covers $\cal U$ of $X$, one can construct a topological space $\Delta X$ as in \cite{Chazal:2008} such that
the following diagram commutes 
\[
\xymatrix{
X \ar@{^{(}->}[r] 
& X' \ar@{^{(}->}[r] 
& X'' \\
\Delta X \ar[u]_{p} \ar@{^{(}->}[r] 
& \Delta X' \ar[u]_{p'} \ar@{^{(}->}[r] 
& \Delta X'' \ar[u]_{p''} 
}
\]
where $p$ and $p'$ are restrictions of $p''$ to $\Delta X$ and $\Delta X'$ respectively, and
$p$,$p'$ and $p''$ are homotopy equivalences.
Therefore, we have a map of pairs $p':(\Delta X', \Delta X) \rightarrow (X',X)$. Considering 
the two long exact sequences of pairs $(\Delta X', \Delta X)$ and $(X',X)$ and using the
same arguments in Proposition \ref{pair-iso}, it follows that $p'_{*}:\HH(\Delta X', \Delta X) \rightarrow \HH(X',X)$ is an
isomorphism. Similarly, $p''$ is also a map of pairs, and the induced homomorphisms $p''_{*}:\HH(\Delta X'', \Delta X) \rightarrow \HH(X'',X)$ is also 
an isomorphism. Given that both $p'$ and $p''$ are maps of pairs, we have the following commutative diagram of pairs: 
\[
\xymatrix{
(X', X) \ar@{^{(}->}[r] 
& (X'',X)  \\
(\Delta X',\Delta X) \ar[u]_{p'} \ar@{^{(}->}[r] 
& (\Delta X'',\Delta X) \ar[u]_{p''}  
}.
\]
It induces the following commutative diagram at homology level where vertical maps are isomorphisms,
and horizontal maps are induced from canonical inclusions.
\[
\xymatrix{
\HH(X', X) \ar[r] 
& \HH(X'',X)  \\
\HH(\Delta X',\Delta X) \ar[u]_{p'_{*}} \ar[r] 
& \HH(\Delta X'',\Delta X) \ar[u]_{p''_{*}}  
}.
\]
Next, let $\Gamma$ be the first barycentric subdivision of ${\cal N}{\cal U}$, and
$\Gamma'$ for ${\cal N}{\cal U}'$ and $\Gamma''$ for ${\cal N}{\cal U}''$, respectively. 
It is shown in \cite{Chazal:2008} that the following diagram commutes 
\[
\xymatrix{
\Delta X \ar[d]^{q} \ar@{^{(}->}[r] 
& \Delta X' \ar[d]^{q'} \ar@{^{(}->}[r] 
& \Delta X'' \ar[d]^{q''} \\
\Gamma \ar@{^{(}->}[r] 
& \Gamma'  \ar@{^{(}->}[r] 
& \Gamma''  
}
\]
where $q$ and $q'$ are the restrictions of $q''$ to $\Delta X$ and $\Delta X'$, respectively.
Following the same arguments as above, one obtains the following commutative diagram at homology level with vertical isomorphisms
and horizontal maps induced from canonical inclusions.
\[
\xymatrix{
\HH(\Delta X', \Delta X) \ar[d]^{q'_{*}} \ar[r] 
& \HH(\Delta X'', \Delta X) \ar[d]^{q''_{*}} \\
\HH(\Gamma', \Gamma) \ar[r] 
& \HH(\Gamma'', \Gamma)  
}
\]
%$\Gamma$, $\Gamma'$ and $\Gamma''$ are first barycentric subdivision of ${\cal N}{\cal U}$, ${\cal N}{\cal U}'$ and ${\cal N}{\cal U}''$, respectively.
It is known that simplicial approximation $g'':\Gamma''\rightarrow {\cal N}{\cal U}''$ of the identity map $id:|\Gamma''| \rightarrow |{\cal N}{\cal U}|$
commutes with canonical inclusions and induces an isomorphism between homology \cite{Munk75}.
Therefore, we have the following commutative diagram
\[
\xymatrix{
\Gamma \ar[d]^{g} \ar@{^{(}->}[r] 
& \Gamma' \ar[d]^{g'} \ar@{^{(}->}[r] 
& \Gamma'' \ar[d]^{g''} \\
{\cal N}{\cal U} \ar@{^{(}->}[r] 
& {\cal N}{\cal U}' \ar@{^{(}->}[r] 
& {\cal N}{\cal U}''  
}
\]
where $g$ and $g'$ are the restriction of $g''$ to $\Gamma$ and $\Gamma'$ respectively.
As before, there exists following commutative diagram at homology level with vertical isomorphisms and horizontal maps induced by
canonical inclusions,
\[
\xymatrix{
 \HH(\Gamma', \Gamma)  \ar[d]^{g'_{*}} \ar[r] 
& \HH(\Gamma'', \Gamma)   \ar[d]^{g''_{*}} \\
\HH({\cal N}{\cal U}', {\cal N}{\cal U}) \ar[r] 
& \HH({\cal N}{\cal U}'', {\cal N}{\cal U}) 
}
\]
Combining these three commutative diagrams at homology level, the first commutative diagram in the proposition follows immediately.
A similar argument shows that the second commutative diagram in the proposition holds as well.
\end{proof}

\vspace*{0.07in}\noindent Using this proposition, one can easily obtain that 
$$\mathrm{im}\left(\HH(\nerv\clX_{\alpha, r}, \nerv\clX_{\alpha, r}^{\lambda+3\delta}) \rightarrow 
\HH(\nerv\clX_{\alpha', r}, \nerv\clX_{\alpha', r}^{\lambda'+\delta'}) \right) \cong  \HH({\M}, {\M} - \bar{p}) $$
because
$$\mathrm{im}\left(\HH(\X_{\alpha, r}, \X_{\alpha, r}^{\lambda+3\delta}) 
\rightarrow \HH(\X_{\alpha', r}, \X_{\alpha', r}^{\lambda'+\delta'})\right) \cong \HH({\M}, {\M} - \bar{p}). $$ 
Indeed, for convenience, set $A_1 =\X_{\alpha, r}$, $B_1 =\X_{\alpha, r}^{\lambda+3\delta}$, $A_2 = \X_{\alpha', r}$ and $B_2 = \X_{\alpha', r}^{\lambda' +\delta'}$;
and set ${\cal A}_1 =\clX_{\alpha, r}$, ${\cal B}_1 =\clX_{\alpha, r}^{\lambda+3\delta}$, ${\cal A}_2 = \clX_{\alpha', r}$ and ${\cal B}_2 = \clX_{\alpha', r}^{\lambda' +\delta'}$. We apply the above proposition twice, once to the three spaces $B_1\subset A_1 \subset A_2$, and once to the three spaces $B_1 \subset B_2 \subset A_2$. This provides the following diagram, where the commutativity of each square follows from Proposition \ref{nerve-relative}. 
\[
\xymatrix{
\HH(A_1, B_1) \ar[d] \ar[r] & \HH(A_2, B_1) \ar[d] \ar[r] & \HH(A_2, B_2) \ar[d] \\
\HH(\nerv{\cal A}_1, \nerv{\cal B}_1) \ar[r] & \HH(\nerv{\cal A}_2, \nerv{\cal B}_1) \ar[r] & \HH(\nerv{\cal A}_2, \nerv{\cal B}_2) 
}
\]
Since all vertical homomorphisms are isomorphisms, we have that 
$$\mathrm{im}\left(\HH(A_1, B_1) \rightarrow \HH(A_2, B_2)\right) \cong \mathrm{im}\left(\HH(\nerv{\cal A}_1, \nerv{\cal B}_1), \HH(\nerv{\cal A}_2, \nerv{\cal B}_2)\right). $$
This finishes the proof of Lemma \ref{lem:relativenerves}. 

\section{The Algorithm to Compute $\mathrm{im}({{j_\alpha}_*})$}
\label{appendix:matrix-algorithm}
Recall that $j_\alpha$ is the inclusion of pairs  
$j_\alpha: (A_1, B_1) \hookrightarrow (A_2,B_2)$,
where $B_1 \subset A_1 \subset A_2$ and $B_1 \subset B_2 \subset A_2$.
To compute $\mathrm{im}({{j_\alpha}_*})$, 
we order the simplices of $A_2$ in a proper way to 
build a filtration such that  
the rank of $\mathrm{im}({j_\alpha}_*)$
can be read off from the reduced boundary matrix built from the filtration. 
Precisely, the filtration adds the simplices of $A_2$ as follows.
The simplices in $B_2\setminus A_1$ appear first.
Then the simplices in $B_1$, $(B_2\setminus B_1)\cap A_1$, 
$A_1 \setminus B_2$ and $A_2 \setminus (A_1\cup B_2)$ follow sequentially.
This order is illustrated in Figure \ref{fig:relative_homology}.
For simplicity, let $\mathrm{R}({\sf x}, {\sf y})$ denote submatrix 
occupying the 
rectangle region with ${\sf x}$ as its top left corner point
and ${\sf y}$ as its bottom right corner point in 
Figure \ref{fig:relative_homology}.
\begin{figure}[ht!]
\begin{center} 
\includegraphics[scale=0.5]{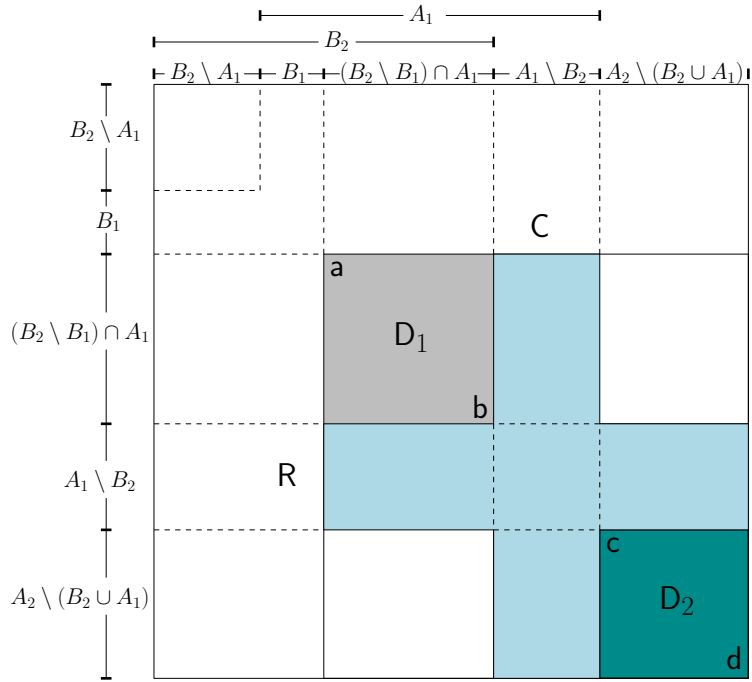} 
\end{center}
\caption{The order of the simplices in the filtration for $A_2$.
} 
\label{fig:relative_homology}
\end{figure}
It is known \cite{EH10} that the rank of $\HH(A_1, B_1)$ 
(or $\HH(A_2, B_2)$) can be 
computed by reducing the submatrix ${ M}_{1} = \mathrm{R}({\sf a}, {\sf c})$ 
(or ${ M}_{2} = \mathrm{R}({\sf b}, {\sf d})$) 
in Figure \ref{fig:relative_homology}.
For our purpose, 
the submatrix ${ M}= \mathrm{R}({\sf a}, {\sf d})$  
in Figure \ref{fig:relative_homology},
which  contains 
both ${ M}_{1}$ and ${ M}_{2}$, 
will be
reduced in the same way as the classical persistent homology algorithm does
\cite{EH10}.
Let $\widehat{ M}$ denote the matrix reduced from ${ M}$.
It will be shown that  
the rank of $\mathrm{im}({j_\alpha}_*)$  can be read off 
from $\widehat{ M}$. 

Recall that the $\mathrm{im}({j_\alpha}_*)$ in dimension $k$ contains the $k$-cycles
of $\HH_k(A_1, B_1)$  which are nontrivial in both 
$\HH_k(A_1, B_1)$ and $\HH_k(A_2, B_2)$.
In particular, each  $k$-simplex in the 
collection of simplices $A_1 \setminus B_2$ whose column in the reduced matrix $\widehat{{ M}}$
is a zero column 
(i.e., a zero column corresponding to a $k$-simplex
in the light blue column region ${\sf C}$ of Figure \ref{fig:relative_homology})
represents a $k$-cycle in both $\HH_k(A_1, B_1)$ and $\HH_k(A_2, B_2)$.
Let $\#Zero_k$ denote the number of such zero columns in $\sf C$.
If one such $k$-simplex is paired by a $(k+1)$-simplex in $A_2 \setminus B_2$
(i.e., the row in  
$\widehat{{ M}}$ corresponding to this simplex which is in 
the light blue row region $\sf R$ of Figure  \ref{fig:relative_homology}
has a unique $1$ ),
its corresponding $k$-cycle is a $k$-boundary in  $\HH_k(A_2, B_2)$.
Let $\#Bdry_k$ denote the number of such $k$-simplices.
Since the $k$-cycles in $\HH_k(A_1, B_1)$ corresponding to
zero columns which appear  before the columns in $\sf C$ 
contain only simplices from $B_2$, 
they all have trivial image in $\HH_k(A_2, B_2)$.
It is then immediate that the rank of $\mathrm{im}({j_\alpha}_*)$ in dimension $k$ 
equals $\#Zero_k - \#Bdry_k$,
namely the number of  
zero columns in $\sf C$ which correspond to unpaired $k$-simplices.
Once the matrix ${ M}$ is reduced, 
it is straightforward to compute $\#Zero_k - \#Bdry_k$.
If there are $n$ simplices in $A_2\setminus B_2$, this algorithm 
runs in $O(n^{3})$ time due to the reduction of 
${ M}$. 

\section{Graph Based Central Points and Experimental Details on Real Data }
\label{appendix:centers}
A graph on sample points is built by connecting two points 
within certain distance.
For every vertex $v$ of each component of this graph, 
the shortest path tree with root $v$ is computed and then the largest distance 
from $v$ to leaves of this shortest path
tree  is recorded. The vertex whose distance to leaves of its 
shortest path tree is 
the minimum among those vertices in the component containing it, 
is considered to 
be the center of its component. 
Intuitively, these centers are away from the boundary 
and less likely to be outliers. 
We then discard the centers of 
components with few points. 
For remaining centers, we compute the local homology and report the intrinsic 
dimension of the manifold as that of the $n$-sphere whose homology is the 
same as the most common local homology of these centers. 
To accelerate the computation,
if a component has a 
significantly large number of vertices, we generate a uniform sparse subsample
from the points within some radius of its center and then compute local homology
on the subsample points.

%\begin{wrapfigure}{l}{3.5in}
\begin{table}[htb]
\centering
\begin{tabular}{|c|c|c|c|c|}
\hline
& \multicolumn{2}{|c|}{$n$-{\sc sphere}} & {\sc Est. dim } &  {\sc percentage} \\
\hline
\textbf{Head}&  n=3 & $53/53$ &  $3$ &  $100\% (53 / 53)$ \\
\hline
\multirow{3}{*}{\textbf{D1}} & 
n=3 & $4/37$  & \multirow{3}{*}{$4$ }  &\multirow{3}{*}{$83.7\% (31 / 37)$}  \\
&  n=4& $31/37$ &  &\\
&  n=5& $2/37$ &  &\\
\hline
\multirow{2}{*}{\textbf{D0}} & 
n=2 & $2/9$  & \multirow{2}{*}{$3$ }  & \multirow{2}{*}{$77.7\% (7 / 9)$}  \\
& n=3& $7/9$ & &  \\
\hline
\end{tabular}%\caption{Comparison results. The $2$nd-$7$th rows are from \cite{CC09}}
\caption{Estimated dimension for real data}
\label{real-tab}
\end{table}
%\end{wrapfigure}
We applied this strategy on \textbf{Head}, \textbf{D1} and \textbf{D0}.
All of them have only one major component in the graph which 
connects two points within a distance that
is several times the 
distance of the closest pair in the sample points. 
%There are $696$ points  in the major component of \textbf{Head},
%$4262$ points in that of \textbf{D1} and  
%$5914$ points in that of \textbf{D0}.
For \textbf{Head}, a subsample of around $138$ points was taken 
from $505$ points in the neighborhood of the center of the major component.
We took a subsample of around $148$ points from $943$ points 
in the neighborhood of 
of \textbf{D1}'s center, and around $102$ points from $3494$ points 
in the neighborhood of \textbf{D0}'s center.
Since the uniform subsamples were taken randomly, 
one will be biased to claim the result from one particular subsample.
Therefore, we repeated the local homology computation at the center with fixed parameters $100$ times.
Note that the points in the subsamples changed each time due to random sampling.
Among these $100$ computations, we only counted the valid ones which returned the local homology 
of $\redHH(\sphere^{n})$ for some $n$. 
The distribution of valid computations is shown in Table~\ref{real-tab}.
The $n$-{\sc sphere} column shows the number of valid computations 
with the reduced homology of $\redHH(\sphere^{n})$
for each $n$. The total number of valid computations is also included in this column.
The {\sc est. dim} column gives the estimated dimension.
The {\sc percentage} column shows the percentage of computations with the estimated dimension
in all valid computations.
For the \textbf{Head} data, the detected dimension from our method
matches the ground truth which is  $3$.
%Ours falls into this range. 
Although the ground truth dimensions 
for \textbf{D1} and 
\textbf{D0} are unknown, ours along with SLIVER, PN, LPCA and ISOMAP report
dimension in range $[3,7]$ for  \textbf{D1} and in range $[2,9]$ for \textbf{D0}.

\end{document}